\documentclass[a4paper, USenglish]{lipicsoid-2019}

\hideLIPIcs
\usepackage{etex}

\usepackage{graphicx} 
\usepackage[usenames,dvipsnames]{pstricks} 
\usepackage{epsfig} 
\usepackage{tikz} 
\usetikzlibrary{intersections}
\usetikzlibrary{shapes} 
\usepackage{rotating} 
\usepackage{tcolorbox} 
\usepackage[linesnumbered,vlined,english,algoruled,noend]{algorithm2e}
\usepackage[lowercase]{theoremref} 
\usepackage{multicol}

\usepackage{fixltx2e}

\usepackage[normalem]{ulem} 

\input{rounding-rectangle-trick.tex}

\usepackage{hyperref} 
\usepackage[numbers,sort&compress]{natbib} 
\usepackage{doi} 
\usepackage{bibentry}
\nobibliography*
\usepackage[strings]{underscore}

\hypersetup{
  colorlinks=true,
  linkcolor=black,
  citecolor=black,
  filecolor=black,
  urlcolor=[rgb]{0,0.1,0.5},
  pdftitle={},
  pdfauthor={}
}

\nolinenumbers

\newtheorem{Theo}{Theorem}
\newtheorem{Claim}{Claim}
\newtheorem{Fa}{Fact}
\newtheorem{Def}{Definition}
\newtheorem{Prop}{Property}

\newtheorem{Coro}{Corollary}
\newtheorem{Lem}{Lemma}

\newcommand{\CC}{\mathcal{C}}
\newcommand{\FF}{\mathcal{F}}

\newtcolorbox{open}[1]{colback=orange!30,colframe=orange!50, fonttitle=\bfseries,title=#1}

\newtcolorbox{combox}[1]{colback=orange!30,colframe=orange!50, fonttitle=\bfseries,title=#1}

\newcommand{\pnograph}{\textsc{No Graph}}
\newcommand{\ppath}{\textsc{Linear Forest}}
\newcommand{\pstar}{\textsc{Star}}
\newcommand{\pinterval}{\textsc{Interval}}
\newcommand{\psplit}{\textsc{Split}}
\newcommand{\pforest}{\textsc{Forest}}
\newcommand{\pbipartite}{\textsc{Bipartite}}
\newcommand{\pchordal}{\textsc{Chordal}}
\newcommand{\pcomparability}{\textsc{Comparability}}
\newcommand{\ptriangle}{\textsc{Triangle-Free}}

\newcommand{\mirror}{\text{{\normalfont mirror-}}}
\newcommand{\co}{\text{{\normalfont co-}}}

\newcommand{\esp}{{\normalfont \&}}

\newcommand{\family}[3]{{[#1]\textsubscript{\textcolor{gray}{#2}}: #3.}}

\title{Graph classes and forbidden patterns on three vertices}

\author{Laurent Feuilloley}
{DII, Universidad de Chile}%
{feuilloley@dii.uchile.cl}%
{https://orcid.org/0000-0002-3994-0898}%
{Additional funding: ANR DESCARTES, ANR ESTATE, Inria GANG and DELYS, MIPP and Jos\'e Correa's Amazon Research award.}

\author{Michel Habib}
{IRIF, Paris University}%
{habib@irif.fr}%
{}%
{Additional funding: ANR HOSIGRA, Inria GANG}

\authorrunning{L. Feuilloley and M. Habib}

\begin{document}
\maketitle{}

\begin{abstract}
 
This paper deals with the characterization and the recognition of graph classes. 
A popular way to characterize a graph class is to list a minimal set of forbidden induced subgraphs. 
Unfortunately, this strategy hardly ever leads to a very efficient recognition algorithm. 
On the other hand, many graph classes can be efficiently recognized by techniques that use some ordering of the nodes, such as the one given by a  traversal. 

We study specifically graphs that have an ordering avoiding some ordered structures. 
More precisely, we consider structures that we call \emph{patterns on three nodes}, and the complexity of recognizing the classes associated with such patterns. In this domain, there are three key previous works. 
In~\cite{Skrien82} and~\cite{Damaschke90} (independently), Skrien and Damashke noted that several graph classes such as chordal, bipartite, interval and comparability graphs have a characterization in terms of forbidden patterns.
On the algorithmic side, Hell, Mohar and Rafiey proved that any class defined by a set of forbidden  patterns on three nodes can be recognized in time $O(n^3)$, using an algorithm based on 2-SAT~\cite{HellMR14}. 
We improve on these two lines of works, by characterizing systematically all the classes defined by sets of forbidden patterns (on three nodes), and proving that among the 22 different classes (up to complement) that we find, 20 can actually be recognized in linear time. 

Beyond these results, we consider that this type of characterization is very useful from an algorithmic perspective, leads to a rich structure of classes, and generates many algorithmic and structural open questions worth investigating.  
\end{abstract}

\newpage{}



\section{Introduction}
\paragraph*{Forbidden structures in graph theory.} 
A class of graphs is hereditary if for any graph~$G$ in the class, every induced subgraph of~$G$ also belongs to the class. 
Given a hereditary class~$\CC$, there exists a family~$\FF$ of graphs, such that a graph~$G$ belongs to~$\CC$, if and only if, $G$ does not contain any graph of~$\FF$ as an induced subgraph.
Hence hereditary classes are defined by forbidden structures. 
For a given class $\CC$, a trivial family~$\FF$ is the set of all graphs not in $\CC$, but the interesting families are the minimal ones. 
If we were to replace the induced subgraph relation by the minor relation, a celebrated theorem of Robertson and Seymour states that these families are always finite, but here the family needs not be finite, as exemplified by bipartite graphs (where the set of forbidden structures is the set of odd cycles). 
There exist many characterizations of classes by forbidden subgraphs in the literature, ranging from easy to extremely difficult to prove, as for example the Strong Perfect Graph Theorem  \cite{ChudnovskyRST06}.

\paragraph*{Forbidden ordered structures.}
Another way of defining hereditary classes by forbidden structures is the following. 
Consider a graph~$H$ given with a fixed ordering on its vertices. 
Then a graph~$G$ belongs to the class associated with~$H$, if and only if, there exists an ordering of the vertices of~$G$, such that none of its subgraphs induces a copy of~$H$ with the given ordering. 
Let us illustrate such characterization with chordal graphs. 
Chordal graphs are usually defined as the graphs that do not contain any induced cycle of length at least $4$. 
However it is also well known \cite{Dirac61} that chordal graphs are exactly the graphs that admit a simplicial elimination ordering, that is an ordering on the vertices such that the neighbors of a vertex that are placed before it in the ordering induce a clique. 
In the framework described above, this characterization corresponds to $H$ being the path on $3$ vertices with the middle vertex placed last in the ordering. 
Similar characterizations are known for well-studied classes such as proper interval, interval, permutation and cocomparability graphs. 

We focus on such characterizations, and refer to forbidden ordered induced subgraphs (or more precisely an equivalent trigraph version of it) as \emph{patterns}.

\paragraph*{Forbidden patterns and algorithms.}

A motivation to study characterizations by patterns is that they are related to efficient recognition algorithms. 
That is, unlike forbidden subgraph characterizations, forbidden patterns characterization often translate into fast algorithms to recognize the class, most of the time \emph{linear-time} algorithms. 
Such algorithms compute an ordering avoiding the forbidden patterns, which means that in addition of deciding whether the graph is in the class or not, they provide a certificate for the positive cases.  
The most famous examples are for chordal graphs~\cite{TarjanY84}, proper interval graphs~\cite{Corneil04}, and interval graphs~\cite{CorneilOS09}. 

These algorithms are fast because they use simple graph searches, in particular the lexicographic breadth first search, LexBFS. 
Indeed the orderings given by searches have a special structure~\cite{CorneilK08}, that matches the one of forbidden patterns.

\paragraph*{Previous works on forbidden patterns.} 
The papers \cite{Skrien82} and \cite{Damaschke90} by Skiren and Damaschke respectively are, as far as we know, the first works to consider characterizations by patterns as a topic in itself.
From these seminal papers, one can derive that all the classes defined by one forbidden pattern on three node can be recognized in polynomial time. 
This is for example the case for the chordal graphs mentioned above.
More recently, it was proved that forbidding \emph{a set} of patterns on three nodes, still leads to a polynomially solvable problem~\cite{HellMR14}. This stands in striking contrast with the case of larger patterns, as it was shown in~\cite{DuffusGR95} that almost all classes defined by 2-connected patterns are NP-complete to recognize. In~\cite{HellMR14}, the authors conjectured a dichotomy on this type of recognition problem, but this statement has been recently challenged~\cite{Nesetril17}.  

More generally, forbidden ordered structures have been used recently in a variety of contexts.
Among others, they have been applied to characterize graphs with bounded asteroidal number~\cite{CorneilS15}, to study intersection graphs~\cite{Wood04}, to prove that the square of the line-graph of a chordal graph is chordal~\cite{BrandstadtH08}, and to study maximal induced matching algorithms~\cite{HabibM17}.

\paragraph*{Our results}
Forbidding only one pattern on three nodes already gives rise to a rich family of graph classes.  
Yet, despite the general algorithmic result of \cite{HellMR14}, little is known about the classes defined by \emph{a set} of patterns on three nodes. 
Our main contribution is an exhaustive list of all the classes defined this way. 
Along the way several interesting results and insights are gathered.
A corollary of this characterization is that almost all the classes considered can be recognized not only in polynomial time, but in linear time.
Beyond these technical contributions, our goal is to unify many results scattered in the literature, and show that this formalism can be useful and relevant. In this sense, the paper also serves as a survey of this type of hereditary classes of graphs.

\paragraph*{Outline of the paper.}

The paper is organized as follows. 
In Section~\ref{sec:appetizer}, we motivate the study of patterns by listing the well-known classes characterized by one pattern on three node.
In Section~\ref{sec:definitions}, we formally define the pattern characterizations, and prove structural properties about the classes defined this way.   
Section~\ref{sec:main-theorems} contains the main theorem of the paper, that is the complete characterization of all the classes defined by sets of patterns on three nodes.
The proof of this theorem is a long case analysis, and to make it nicer, we first highlight and prove some remarkable characterizations in Section~\ref{sec:special}, before completing the proof in Section~\ref{sec:main-proof}.
In Section~\ref{sec:algorithms}, we deal with the algorithmic aspects, in particular the linear-time recognition. 
Finally we discuss related topics and open questions in Section~\ref{sec:discussions}.


\section{Appetizer: classes defined by one pattern}
\label{sec:appetizer}

In this section, we introduce a few classes characterized by one pattern on three nodes. 
We do so before defining formally pattern characterizations, as an appetizer for the rest of the paper. 
All these characterizations are known, and in particular they are listed in \cite{Damaschke90}.\footnote{For completeness we prove these characterizations in Section~\ref{sec:special}, as no proof can be found in \cite{Damaschke90}.}
The main take-away of this short section is that essential graph classes can be defined by one pattern on three nodes, hence such characterizations are important.

Let us introduce a graphic intuitive representation of patterns.
Consider the drawing of Figure~\ref{fig:graphic-interval}.

\begin{figure}[!h]
\begin{center}
\scalebox{1.8}{
\begin{tikzpicture}
			[scale=0.7,auto=left,every node/.style=		
			{circle,draw,fill=black!5}]
			\node (a) at (0,0) {};
			\node (b) at (1,0) {};
			\node (c) at (2,0) {};
			\draw (a) to[bend left=50] (c);
			\draw[dashed] (a) to (b);
		\end{tikzpicture}}
\end{center}
\caption{\label{fig:graphic-interval} The graphic representation of the pattern associated with the class of interval graphs.}
\end{figure}
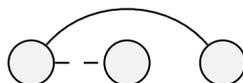

The class associated with this pattern is the class of graphs $G=(V,E)$ that have an ordering of their vertices such that: there is no ordered triplet of nodes, $a<b<c$, such that $(a,b)\notin E$ and $(a,c)\in E$. 
In other words, the forbidden configuration consists in a non-edge (dashed edge in the drawing) between the two first nodes, an edge (plain edge in the drawing) between the first and the last node, and there is no constraint (no edge in the drawing) on the edge between the second and the third node (that is, whether there is an edge or not, the configuration forbidden). 

\begin{Theo}(\cite{Damaschke90})\label{thm:Damaschke}
In Table~\ref{tab:one-pattern-3}, each class of the first column is characterized by the pattern  represented in the second column. (The third column is the name we give to this pattern.) 
\end{Theo}

\begin{table}[!h]
\begin{center}
\fbox{
\begin{tabular}{ccc}
Graph class & Pattern representation & Pattern name\\
\hline
\\[-0.15cm]
Linear forests & \begin{tikzpicture}[scale=0.7,auto=left, every node/.style=	{circle, draw, fill=black!5}]
	\node (a) at (0,0) {};
	\node (b) at (1,0) {};
	\node (c) at (2,0) {};
	\draw (a) to[bend left=50] (c);
\end{tikzpicture} & \ppath{}\vspace{0.27cm}\\
Stars & \begin{tikzpicture}[scale=0.7,auto=left,every node/.style=	 {circle, draw, fill=black!5}]
	\node (a) at (0,0) {};
	\node (b) at (1,0) {};
	\node (c) at (2,0) {};
	\draw (a) to (b);

\end{tikzpicture}& \pstar{}\\
Interval graphs & \begin{tikzpicture}
			[scale=0.7,auto=left,every node/.style=		
			{circle,draw,fill=black!5}]
			\node (a) at (0,0) {};
			\node (b) at (1,0) {};
			\node (c) at (2,0) {};
			\draw (a) to[bend left=50] (c);
			\draw[dashed] (a) to (b);
		\end{tikzpicture}& \pinterval{} \vspace{0.27cm}\\
Split graphs& \input{split.tex}& \psplit{}\\
Forest & \begin{tikzpicture}
			[scale=0.7,auto=left,every node/.style=		
			{circle,draw,fill=black!5}]
			\node (a) at (0,0) {};
			\node (b) at (1,0) {};
			\node (c) at (2,0) {};
			\draw (a) to[bend left=50] (c);
			\draw (b) to (c);
		\end{tikzpicture}& \pforest{}\vspace{0.27cm}\\
Bipartite graphs & \begin{tikzpicture}
	[scale=0.7,auto=left,every node/.style=		
	{circle,draw,fill=black!5}]
	\node (a) at (0,0) {};
	\node (b) at (1,0) {};
	\node (c) at (2,0) {};
	\draw (a) to (b);
	\draw (b) to (c);
\end{tikzpicture}& \pbipartite{}\\
Chordal graphs& \begin{tikzpicture}
	[scale=0.7,auto=left,every node/.style=		
	{circle,draw,fill=black!5}]
	\node (a) at (0,0) {};
	\node (b) at (1,0) {};
	\node (c) at (2,0) {};
	\draw (a) to[bend left=50] (c);
	\draw[dashed] (a) to (b);
	\draw (b) to (c);
\end{tikzpicture}& \pchordal{}\\
Comparability graphs & \begin{tikzpicture}
			[scale=0.7,auto=left,every node/.style=		
			{circle,draw,fill=black!5}]
			\node (a) at (0,0) {};
			\node (b) at (1,0) {};
			\node (c) at (2,0) {};
			\draw[dashed] (a) to[bend left=50] (c);
			\draw (a) to (b);
			\draw (b) to (c);
		\end{tikzpicture}& \pcomparability\\
Triangle-free graphs & \begin{tikzpicture}
			[scale=0.7,auto=left,every node/.style=		
			{circle,draw,fill=black!5}]
			\node (a) at (0,0) {};
			\node (b) at (1,0) {};
			\node (c) at (2,0) {};
			\draw (a) to[bend left=50] (c);
			\draw (a) to (b);
			\draw (b) to (c);
		\end{tikzpicture}& \ptriangle{}\\[0.27cm]
At most two nodes & \begin{tikzpicture}[scale=0.7,auto=left, every node/.style=	{circle, draw, fill=black!5}]
	\node (a) at (0,0) {};
	\node (b) at (1,0) {};
	\node (c) at (2,0) {};
\end{tikzpicture}& \pnograph 
\end{tabular}
}
\vspace{0.5cm}
\caption{\label{tab:one-pattern-3} The table of Theorem~\ref{thm:Damaschke}. For each row the class on the first column is characterized by the pattern represented in the second column, whose name is on the third column.}
\end{center}
\end{table}


\section{Definitions and structural properties}
\label{sec:definitions}

The current paper aims at doing a thorough study of the classes defined by patterns on three nodes, and of their relations. 
In this section we give the formal definitions and introduce structural properties of these classes. 
More precisely, we start in Subsection~\ref{subsec:def} by defining formally the  patterns and related concepts, and in Subsection~\ref{subsec:graph-classes} by listing the definitions of the graph classes we will use.
Then in Subsection~\ref{subsec:operations} we define some basic operations on patterns and in Subsections~\ref{subsec:basic}, \ref{subsec:pattern-split}, and~\ref{subsec:union-intersection}, we describe important structural properties of the classes defined by patterns.

\subsection{Definitions related to patterns}
\label{subsec:def}

In all the paper we deal with finite loopless undirected graphs, and multiple edges are not allowed.
For such a graph $G$, we denote by $V(G)$ the set of vertices and $E(G)$ its set of edges, with the usual $|V(G)|=n$ and $|E(G)|=m$ for complexity evaluations.
The graph we consider are not necessarily connected.

To define patterns, we use the vocabulary of trigraphs as, for example, in \cite{Chudnovsky06}.

\begin{Def}
A \emph{trigraph} $T$ is a $4$-tuple $(V(T),E(T),N(T),U(T))$ where $V(T)$ is the vertex set and every unordered pair of vertices belongs to one of the three disjoint sets $E(T)$, $N(T)$, and $U(T)$, called the \emph{edges}, \emph{non-edges} and \emph{undecided edges}, respectively. 
A graph $G=(V(G),E(G))$ is a \emph{realization} of a trigraph $T$ if $V(G)=V(T)$ and $E(G)=E(T)\cup U'$, where $U'\subset U(T)$.
\end{Def}

When representing a trigraph, we will draw plain lines for edges, dashed lines for non edges, and nothing for undecided edges. 
Also as $(E,N,U)$ is a partition of the unordered pairs, it is enough to give any two of these sets to define the trigraph, and we will often define a trigraph by giving only $E$ and $N$.

\begin{Def}
An \emph{ordered graph} is a graph given with a total ordering of its vertices. A \emph{pattern} is an ordered trigraph. An ordered graph is a \emph{realization} of a pattern if they have the same set of vertices, with the same linear ordering, and the graph is a realization of the trigraph. 
When, in an ordered graph, no ordered subgraph is the realization of given pattern, the ordered graph \emph{avoids} the pattern.
\end{Def}

In this formalism, the pattern for interval graphs represented in Figure~\ref{fig:graphic-interval} is $(E,N)=(\{(1,3)\},\{(1,2)\})$. 
For simplicity we give names to the patterns, related to the classes they characterize. These names are in capital letters to avoid confusion with the classes. For example the one of Figure~\ref{fig:graphic-interval} is called \pinterval{}.
More generally, the names follow from Theorem~\ref{thm:Damaschke}.
The list of the names of all the patterns is given in Figure~\ref{fig:27-patterns}. The pattern that have no undecided edge are called \emph{full patterns}.

\begin{Def}
Given a family of patterns $\mathcal{F}$, the class $\CC_{\FF}$ is the set of connected graphs that have the following property: there exists an ordering of the nodes that avoids all the patterns in $\mathcal{F}$. 
\end{Def}

To make these definitions more concrete, let us find out which classes are characterized by patterns on two nodes, $V=\{1,2\}$. 
Forbidding the pattern $(E,N)=(\{(1,2)\},\emptyset)$  means that the graphs we consider have a vertex ordering such that there is no pair of nodes $a<b$ with an edge $(a,b)$. 
This implies that the graph has actually no edge thus this is the class of independent sets.
Similarly forbidding the pattern $(E,N)=(\emptyset,\{(1,2)\})$ leads to the cliques.
Finally $(E,N)=(\emptyset,\emptyset)$ corresponds to a trivial class: only the graph with one node does not have two nodes that are either linked or not by an edge. 

Finally, if $\mathcal{F}$ consists of only one pattern $P$, we  write $P$ instead of $\{P\}$, thus $\CC_P$ instead of~$\CC_{\{P\}}$.

\subsection{Operations on patterns and families}
\label{subsec:operations}
We define a few operations on patterns and pattern families.

\begin{Def} The \emph{mirror} and \emph{complement} operations are the following:
\begin{itemize}
\item The \emph{mirror of a pattern} is the same pattern, except for the vertex ordering, which  is reversed. The \emph{mirror of a family} $\FF$ is the set of the mirrors of the patterns of the family, and is denoted by \mirror$\FF$.
\item The \emph{complement of a pattern} $(V,E,N)$ is the pattern $(V',E',N')$ with $V'=V$, $E'=N$, and $N'=E$, that is, the pattern where the edges and non-edges have been exchanged. The \emph{complement of a  family} $\FF$ is the set of the complements of the patterns of the family, and is denoted by \co$\FF$.  
\item A pattern $P_2$ is an \emph{extension} of a pattern $P_1$, if it can be obtained by taking $P_1$, and having the possibility to add nodes and to decide undecided edges. A family $\FF_2$ \emph{extends a family} $\FF_1$, if every pattern of $\FF_1$ has an extension in $\FF_2$, and every pattern in $\FF_2$ is an extension of a pattern in $\FF_1$. 
\end{itemize} 
\end{Def}

\subsection{Basic structural properties}
\label{subsec:basic}
We now list some basic structural properties of the classes defined by forbidden patterns. Most of them also appear in \cite{Damaschke90}.
We omit the proofs, as they follow directly from the definitions.

\begin{Prop}\label{prop:basic}
The following properties hold for any pattern family $\FF$.

\begin{enumerate}
\item (Vertex closure) \label{item:closure} The class $\CC_\FF$ is closed under vertex deletion, that is, is hereditary.
\item (Edge Closure) If for every pattern $P$ of the family, $N(P)=\emptyset$, then the class is closed under edge deletion.
\item(Mirror) \label{item:mirror} The family $\mirror\FF$ defines the same graph class as $\FF$, that is, ${\CC_{\mirror\FF}=\CC_\FF}$.
\item (Exchange--Complement) The class  $\CC_{\co\FF}$ defines the complement class of $\CC_\FF$. 
\item (Union) \label{item:union}Given two families $\FF_1$ and $\FF_2$,  $\CC_{\FF_1 \cup \FF_2} \subseteq \CC_{\FF_1} \cap  \CC_{\FF_2}$.
\item (Extension) \label{item:extension} If a family $\FF_2$ extends a family $\FF_1$ then $\CC_{\FF_1} \subseteq \CC_{\FF_2}$.
\end{enumerate}
\end{Prop}
%

Item~\ref{item:closure} states that the classes defined by patterns are hereditary. The converse is also true (if we allow infinite families). 
Indeed any hereditary family has a characterization by a family of forbidden induced subgraphs, and such characterization can automatically be translated into a characterization by patterns: just take the union of all the orderings of all the forbidden subgraphs. 


\subsection{Pattern names}
\label{subsec:pattern-names}

In order to designate patterns in an efficient way, we gave them names and numbers described in Figure~\ref{fig:27-patterns}. The names are inspired by Theorem~\ref{thm:Damaschke} and the basic operations of Subsection~\ref{subsec:basic}.
\begin{figure}[!h]
\hspace{-0.6cm}
\scalebox{0.7}{

\begin{tabular}{ccc}
\begin{tabular}{cc}
0: \ptriangle{}
&
\begin{tikzpicture}
	[scale=1,auto=left,every node/.style=		
	{circle,draw,fill=black!5}]
	\node (a) at (0,0) {};
	\node (b) at (1,0) {};
	\node (c) at (2,0) {};
	\draw (a) to (b);
	\draw (a) to[bend left=50] (c);
	\draw (b) to (c);
\end{tikzpicture}\\
1: mirror-\pchordal{}
&
\begin{tikzpicture}
	[scale=1,auto=left,every node/.style=		
	{circle,draw,fill=black!5}]
	\node (a) at (0,0) {};
	\node (b) at (1,0) {};
	\node (c) at (2,0) {};
	\draw (a) to (b);
	\draw (a) to[bend left=50] (c);
	\draw[dashed] (b) to (c);
\end{tikzpicture}\\
2: \pcomparability
&
\begin{tikzpicture}
	[scale=1,auto=left,every node/.style=		
	{circle,draw,fill=black!5}]
	\node (a) at (0,0) {};
	\node (b) at (1,0) {};
	\node (c) at (2,0) {};
	\draw (a) to (b);
	\draw[dashed] (a) to[bend left=50] (c);
	\draw (b) to (c);
\end{tikzpicture}\\
3 co-\pchordal{}
&
\begin{tikzpicture}
	[scale=1,auto=left,every node/.style=		
	{circle,draw,fill=black!5}]
	\node (a) at (0,0) {};
	\node (b) at (1,0) {};
	\node (c) at (2,0) {};
	\draw (a) to (b);
	\draw[dashed] (a) to[bend left=50] (c);
	\draw[dashed] (b) to (c);
\end{tikzpicture}\\
4: \pchordal{}
&
\begin{tikzpicture}
	[scale=1,auto=left,every node/.style=		
	{circle,draw,fill=black!5}]
	\node (a) at (0,0) {};
	\node (b) at (1,0) {};
	\node (c) at (2,0) {};
	\draw[dashed] (a) to (b);
	\draw (a) to[bend left=50] (c);
	\draw (b) to (c);
\end{tikzpicture}\\
5: co-\pcomparability
&
\begin{tikzpicture}
	[scale=1,auto=left,every node/.style=		
	{circle,draw,fill=black!5}]
	\node (a) at (0,0) {};
	\node (b) at (1,0) {};
	\node (c) at (2,0) {};
	\draw[dashed] (a) to (b);
	\draw (a) to[bend left=50] (c);
	\draw[dashed] (b) to (c);
\end{tikzpicture}\\
6: mirror-co-\pchordal{}
&
\begin{tikzpicture}
	[scale=1,auto=left,every node/.style=		
	{circle,draw,fill=black!5}]
	\node (a) at (0,0) {};
	\node (b) at (1,0) {};
	\node (c) at (2,0) {};
	\draw[dashed] (a) to (b);
	\draw[dashed] (a) to[bend left=50] (c);
	\draw (b) to (c);
\end{tikzpicture}\\
7: co-\ptriangle{}
&
\begin{tikzpicture}
	[scale=1,auto=left,every node/.style=		
	{circle,draw,fill=black!5}]
	\node (a) at (0,0) {};
	\node (b) at (1,0) {};
	\node (c) at (2,0) {};
	\draw[dashed] (a) to (b);
	\draw[dashed] (a) to[bend left=50] (c);
	\draw[dashed] (b) to (c);
\end{tikzpicture}\\
\end{tabular}

&

\begin{tabular}{cc}
8: \pforest{}
&
\begin{tikzpicture}
	[scale=1,auto=left,every node/.style=		
	{circle,draw,fill=black!5}]
	\node (a) at (0,0) {};
	\node (b) at (1,0) {};
	\node (c) at (2,0) {};
	\draw (a) to[bend left=50] (c);
	\draw (b) to (c);
\end{tikzpicture}\\
9: mirror-\pinterval{}
&
\begin{tikzpicture}
	[scale=1,auto=left,every node/.style=		
	{circle,draw,fill=black!5}]
	\node (a) at (0,0) {};
	\node (b) at (1,0) {};
	\node (c) at (2,0) {};
	\draw (a) to[bend left=50] (c);
	\draw[dashed] (b) to (c);
\end{tikzpicture}\\
10: mirror-co-\pinterval{}
&
\begin{tikzpicture}
	[scale=1,auto=left,every node/.style=		
	{circle,draw,fill=black!5}]
	\node (a) at (0,0) {};
	\node (b) at (1,0) {};
	\node (c) at (2,0) {};
	\draw[dashed] (a) to[bend left=50] (c);
	\draw (b) to (c);
\end{tikzpicture}\\
11: co-\pforest{}
&
\begin{tikzpicture}
	[scale=1,auto=left,every node/.style=		
	{circle,draw,fill=black!5}]
	\node (a) at (0,0) {};
	\node (b) at (1,0) {};
	\node (c) at (2,0) {};
	\draw[dashed] (a) to[bend left=50] (c);
	\draw[dashed] (b) to (c);
\end{tikzpicture}\\
\\
12: \pbipartite{}
&
\begin{tikzpicture}
	[scale=1,auto=left,every node/.style=		
	{circle,draw,fill=black!5}]
	\node (a) at (0,0) {};
	\node (b) at (1,0) {};
	\node (c) at (2,0) {};
	\draw (a) to (b);
	\draw (b) to (c);
\end{tikzpicture}\\
\\
13: \psplit{}
&
\begin{tikzpicture}
	[scale=1,auto=left,every node/.style=		
	{circle,draw,fill=black!5}]
	\node (a) at (0,0) {};
	\node (b) at (1,0) {};
	\node (c) at (2,0) {};
	\draw (a) to (b);
	\draw[dashed] (b) to (c);
\end{tikzpicture}\\
\\
14: mirror-\psplit{}=co-\psplit{}
&
\begin{tikzpicture}
	[scale=1,auto=left,every node/.style=		
	{circle,draw,fill=black!5}]
	\node (a) at (0,0) {};
	\node (b) at (1,0) {};
	\node (c) at (2,0) {};
	\draw[dashed] (a) to (b);
	\draw (b) to (c);
\end{tikzpicture}\\
\\
15: co-\pbipartite{}
&
\begin{tikzpicture}
	[scale=1,auto=left,every node/.style=		
	{circle,draw,fill=black!5}]
	\node (a) at (0,0) {};
	\node (b) at (1,0) {};
	\node (c) at (2,0) {};
	\draw[dashed] (a) to (b);
	\draw[dashed] (b) to (c);
\end{tikzpicture}\\
16: mirror-\pforest{}
&
\begin{tikzpicture}
	[scale=1,auto=left,every node/.style=		
	{circle,draw,fill=black!5}]
	\node (a) at (0,0) {};
	\node (b) at (1,0) {};
	\node (c) at (2,0) {};
	\draw (a) to (b);
	\draw (a) to[bend left=50] (c);
\end{tikzpicture}\\
17: co-\pinterval{}
&
\begin{tikzpicture}
	[scale=1,auto=left,every node/.style=		
	{circle,draw,fill=black!5}]
	\node (a) at (0,0) {};
	\node (b) at (1,0) {};
	\node (c) at (2,0) {};
	\draw (a) to (b);
	\draw[dashed] (a) to[bend left=50] (c);
\end{tikzpicture}\\
18: \pinterval{}
&
\begin{tikzpicture}
	[scale=1,auto=left,every node/.style=		
	{circle,draw,fill=black!5}]
	\node (a) at (0,0) {};
	\node (b) at (1,0) {};
	\node (c) at (2,0) {};
	\draw[dashed] (a) to (b);
	\draw (a) to[bend left=50] (c);
\end{tikzpicture}\\
19: mirror-co-\pforest{}
&
\begin{tikzpicture}
	[scale=1,auto=left,every node/.style=		
	{circle,draw,fill=black!5}]
	\node (a) at (0,0) {};
	\node (b) at (1,0) {};
	\node (c) at (2,0) {};
	\draw[dashed] (a) to (b);
	\draw[dashed] (a) to[bend left=50] (c);
\end{tikzpicture}\\
\end{tabular}

&

\begin{tabular}{cc}
20: mirror-\pstar{}
&
\begin{tikzpicture}
	[scale=1,auto=left,every node/.style=		
	{circle,draw,fill=black!5}]
	\node (a) at (0,0) {};
	\node (b) at (1,0) {};
	\node (c) at (2,0) {};
	\draw (b) to (c);
\end{tikzpicture}\\
\\
21: mirror-co-\pstar{}
&
\begin{tikzpicture}
	[scale=1,auto=left,every node/.style=		
	{circle,draw,fill=black!5}]
	\node (a) at (0,0) {};
	\node (b) at (1,0) {};
	\node (c) at (2,0) {};
	\draw[dashed] (b) to (c);
\end{tikzpicture}\\
22: \ppath{}
&
\begin{tikzpicture}
	[scale=1,auto=left,every node/.style=		
	{circle,draw,fill=black!5}]
	\node (a) at (0,0) {};
	\node (b) at (1,0) {};
	\node (c) at (2,0) {};
	\draw(a) to[bend left=50] (c);
\end{tikzpicture}\\
23: co-\ppath{}
&
\begin{tikzpicture}
	[scale=1,auto=left,every node/.style=		
	{circle,draw,fill=black!5}]
	\node (a) at (0,0) {};
	\node (b) at (1,0) {};
	\node (c) at (2,0) {};
	\draw[dashed] (a) to[bend left=50] (c);
\end{tikzpicture}\\
\\
24: \pstar{}
&
\begin{tikzpicture}
	[scale=1,auto=left,every node/.style=		
	{circle,draw,fill=black!5}]
	\node (a) at (0,0) {};
	\node (b) at (1,0) {};
	\node (c) at (2,0) {};
	\draw (a) to (b);
\end{tikzpicture}\\
\\
25: co-\pstar{}
&
\begin{tikzpicture}
	[scale=1,auto=left,every node/.style=		
	{circle,draw,fill=black!5}]
	\node (a) at (0,0) {};
	\node (b) at (1,0) {};
	\node (c) at (2,0) {};
	\draw[dashed] (a) to (b);
\end{tikzpicture}\\
\\
\\
\\
26: \pnograph{}
&
\begin{tikzpicture}
	[scale=1,auto=left,every node/.style=		
	{circle,draw,fill=black!5}]
	\node (a) at (0,0) {};
	\node (b) at (1,0) {};
	\node (c) at (2,0) {};
\end{tikzpicture}\\
\end{tabular}

\end{tabular}
}
\caption{\label{fig:27-patterns} The 27 patterns on three nodes. By convention since mirror-\psplit{}=co-\psplit{}, we will ignore the pattern mirror-\psplit{}.}
\end{figure}
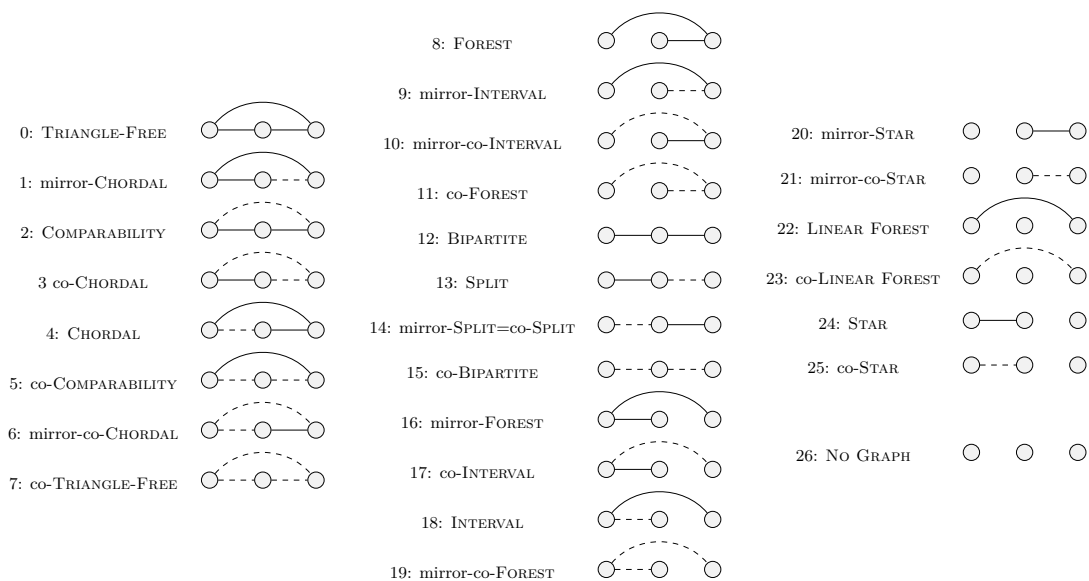

\subsection{Pattern split rule}
\label{subsec:pattern-split}

Let us now describe the \emph{pattern split rule}, which is basically a rewriting rule, that we will use extensively. 
It states that a pattern that has an undecided edge can be replaced by two patterns (that is, can be \emph{split} into two patterns): one where this edge is a (plain) edge, and one where it is a non-edge. 

\begin{Lem}\label{lem:pattern-split}
Let $\FF$ be a pattern family. Let $P=(V,E,N)$ be a pattern of $\FF$, and $e$ be an undecided edge of $P$. 
An ordered graph avoids the patterns of $\FF$, if and only if, it avoids the patterns of $\FF'$, where $\FF'$ is that same as $\FF$ except that $P$ that has been replaced by $P_1=(V,E\cup{e},N)$ and $P_2=(V,E,N\cup{e})$. 
\end{Lem}

\begin{proof}
It is sufficient prove the statement for the case where $\FF$ is restricted to $P$, as the other pattern do not interfere. 
Consider a graph $G$, with an ordering $\tau$.
If $(G,\tau)$ avoids the pattern $P$, then clearly $\tau$ avoids also the patterns $P_1, P_2$, since each occurence of a pattern $P_1$ or $P_2$ yields an occurence of pattern $P$.
Reciprocally, if $(G,\tau)$ avoids the patterns $P_1, P_2$ then it also avoids the pattern $P$, since each possible occurence of pattern $P$ in $(G,\tau)$ must corresponds to either an occurence of pattern $P_1$ or of pattern $P_2$.
\end{proof}

As a consequence, with the notations of the lemma, we have $C_{\FF}=C_{\FF'}$. 
By iterating the rule, one can always transform an arbitrary family of patterns into a family of full patterns. 
Actually the seminal papers on this topic, \emph{e.g.} \cite{Damaschke90}, use only full patterns. 
In this paper, we choose to use undecided edges because they allow for compact notations and provide additional insights. 

The notation $P= P_1 \& P_2$ denotes that $P$ can be split into $P_1$ and $P_2$. 
For example, \pinterval{} = \pchordal{} \& co-\pcomparability{}. 
A family is \emph{split-minimal} if there are no two patterns $P_1$ and $P_2$ in the family, such that there exists a third pattern $P$ with $P=P_1\&P_2$.

\subsection{Union-intersection property}
\label{subsec:union-intersection}

Item \ref{item:union} of Property~\ref{prop:basic} states that $\CC_{\FF_1 \cup \FF2} \subseteq \CC_{\FF_1} \cap  \CC_{\FF_2}$. When the equality holds, we say that these classes have the \emph{union-intersection property}. A trivial case of union-intersection property is when one family is included in the other. 
Here is a more interesting example. Using Lemma~\ref{lem:pattern-split} and Item \ref{item:union} of Property \ref{prop:basic}, we know that: 
\[
\CC_{\pinterval{}}
= \CC_{\pchordal \& \text{co-}\pcomparability}
\subseteq  \CC_{\pchordal} \cap \CC_{\text{co-}\pcomparability}.
\]
But it is known from the literature \cite{GilmoreH64} that the last inclusion is actually an equality: interval graphs are exactly the graphs that are both chordal and cocomparability. 
We will see several other cases where the union-intersection property holds.
From the example above, it is tempting to conjecture that any time the pattern split rule applies, the union-intersection property holds, but this is actually wrong. 
For example \ppath{}=\pforest{} \& \mirror\pinterval{} thus 
$\CC_{\pforest \& \mirror\pinterval}$ is the class of linear forests, but $\CC_{\pforest}\cap \CC_{\mirror\pinterval}$ is a larger class, as it contains any star (more details are given in Property~\ref{prop:not-union-intersection}). 

A useful special case of Item \ref{item:union} in Property~\ref{prop:basic} is the following:

\begin{Fa}\label{fact:inclusion}
Let ${\cal F}_1$, ${\cal F}_2$ be two sets of patterns, if ${\cal C}_{{\cal F}_1} \subseteq {\cal C}_{{\cal F}_2}$, then 
${\cal C}_{{\cal F}_1 \cup {\cal F}_2}\subseteq {\cal C}_{{\cal F}_1}$.
\end{Fa}

It should be noticed that  ${\cal C}_{{\cal F}_1} \subseteq {\cal C}_{{\cal F}_2}$ could derive from a structural graph theorem and not directly from the set of patterns. 
Furthermore, even in this restricted case, the union-intersection property does not always hold, as we will prove in Theorem~\ref{thm:mirror}, that \{\pforest{}, mirror-\pforest{}\} characterizes only paths.

Another special case, is when the patterns are stable by permutation of the ordering of the nodes. On three nodes only the patterns \pnograph{}, \ptriangle{} and co-\ptriangle{} have this property. 
If one of the families at hand contains only pattern with this property, then the union-intersection property holds.

\begin{Fa}\label{fact:order-invariant}
Let ${\cal F}_1$, ${\cal F}_2$ be two sets of patterns, if one of the two families contains only patterns that are stable by change of the ordering of the nodes, then $\CC_{\FF_1 \cup \FF_2}=\CC_{\FF_1}\cap \CC_{\FF_2}$.
\end{Fa}

\subsection{Graph classes}
\label{subsec:graph-classes}

In this section, we define the graph classes we use. 
References on this topic are \cite{Brandstadt1999, graphclasses}. 
In the remaining, we will use $P_i$ to refer to an induced path with $i$~nodes, and $C_i$ refers to an induced cycle on $i$ nodes.

For many classes we give several definitions. In addition to highlighting similarities between some of these classes, this will be helpful in the characterization proofs. For example, characterization by forbidden subgraph are helpful when proving that a set of patterns imply that the graph belong to a class. On the other hand an incremental of geometric construction usually provides a natural ordering to avoid the patterns at hand. 

\begin{Def}\label{def:graph-classes}
The definitions of the main graph classes we use are the following:
\begin{enumerate}
\item A \textbf{\emph{forest}} is a graph with no cycle.
\item A \textbf{\emph{linear forest}} is a disjoint union of paths.
\item A \textbf{\emph{star}} is a graph where at most one node has several neighbors. 
\item An \textbf{\emph{interval graph}} is the intersection graph of a set of intervals. That is, a graph on~$n$~vertices is an interval graph, if there exists a set of $n$~intervals that we can identify to the vertices such that two intervals intersect if and only if the associated vertices are adjacent. 
\item A graph is a \textbf{\emph{split graph}} if there exists a partition of the vertices such that the subgraph induced by the first part is a clique, and the subgraph induced by the second part is an independent set. Note that the split graphs are connected, except (possibly) for isolated nodes.
\item A graph is \textbf{\emph{bipartite}} if there exists a partition of the vertices in 2 parts  such that the 2 resulting induced subgraphs are independent sets.  
\item A graph is \textbf{\emph{chordal}} if it contains no induced cycle of length strictly greater than 3.
\item A graph is a \textbf{\emph{comparability graph}} if its edges represent a partial order. That is, a graph on $n$ vertices is a comparability graph if there exists a partial order with $n$ elements that we can identify with the vertices, such that two elements are comparable if and only if they are adjacent in the graph. 
\item A graph is \textbf{\emph{triangle-free}} if it contains no clique of size 3. 
\item A \textbf{\emph{permutation graph}} (\cite{EvenPL72}) is a graph whose vertices represent the elements of a permutation, and whose edges link pairs of elements that are reversed by the permutation. We will also consider \emph{bipartite permutation graphs}, the subclass of the permutation graphs that are bipartite.
\item A \textbf{\emph{threshold graphs}} is equivalently (\cite{ChvatalH77, MahadevP95}): 
\begin{enumerate}
\item \label{item-def:threshold-increment} 
a graph that can be constructed by incrementally adding isolated vertices and dominating vertices. 
\item \label{item-def:threshold-P4} 
a split graph without induced $P_4$.
\item \label{item-def:threshold-inclusion} 
a split graph where the neighborhoods of the nodes of the independent set and of the nodes of the clique are totally ordered. That is, a graph with vertex set $V=I\cup K$, with $I=\{i_1,...,i_p\}$ an independent set and $K=\{k_1,...,k_q\}$ a clique, such that $N(i_1) \subseteq  \dots \subseteq N(i_p)$  and $N(k_1) \supseteq \dots \supseteq N(k_q)$.
\end{enumerate}
\item A \textbf{\emph{proper interval graph}} is equivalently (\cite{Roberts69}):
\begin{enumerate}
\item the intersection graphs of a set of intervals, where no interval is included in another.
\item a \emph{unit interval graph} , that is interval graphs where all the intervals of the geometric representations have the same length. 
\item an \emph{indifference graph}, that is a graph 
where every node $v$ can be given a real number $k_v$ such that $(u,v)\in E$ if and only if $|k_u-k_v|\leq 1$. 
\end{enumerate}
\item \label{item-def:caterpillar}
A \textbf{\emph{caterpillar graph}} is equivalently:
\begin{enumerate}
\item A forest where each tree has a dominating path.
\item A $(T_2,cycle)$-free graph, where $T_2$ is the graph on seven nodes obtained by taking a 3-star, and appending an additional node on each leaf.
\end{enumerate} 
\item 
A \textbf{\emph{trivially perfect graph}} is equivalently (\cite{Golumbic78, Wolk62}):
\begin{enumerate}
\item a graph in which, for every induced subgraph, the size of a maximum independent set is equal to the number of maximal cliques.
\item \label{item:quasi-threshold} a \emph{quasi-threshold graph}, that is a graph that can be constructed recursively the following way: a single node is a quasi-threshold graph, the disjoint union of two quasi-threshold graphs is a quasi-threshold graph, adding one universal vertex to threshold graph gives a quasi-threshold graph.
\item \label{item:p4-c4} a $(C_4,P_4)$-free graph.
\item \label{item:comparability-of-tree} a comparability graph of an arborescence, that is the comparability graph of a partial order in which for every element $x$, the elements of $\{y|y<x\}$ can be linearly ordered.
\item the intersection graph of a set of nested intervals (that is of intervals such that for every intersecting pair, one interval is included in the other).
\end{enumerate}
\item A \textbf{\emph{bipartite chain graph}} is equivalently (\cite{Yannakakis82}): 
\begin{enumerate}
\item \label{item:bip-chain-neighborhood} a bipartite graph, for which, in each class, one can order the neighborhoods by inclusion. That is, with a partition, $A,B$, $A=a_1, \dots,  a_{|A|}$ satisfies $N(a_1) \supseteq N(a_2), \dots, \supseteq N(a_{|A|})$ and  $B=b_1, \dots,  b_{|B|}$ satisfies $N(b_1) \subseteq N(b_2) \dots \subseteq N(b_{|B|})$.  
\item a \emph{difference graph}, that is a graph 
where every node $v$ can be given a real number $k_v$ in $(-1,1)$ such that $(u,v)\in E$ if and only if $|k_u-k_v|\geq 1$. 
\item \label{item:bip-chain-forbidden} a $2K_2$-free bipartite graph, that is a bipartite graph that does not have two independent edges that are not  linked by a third edge (\emph{i.e.} no induced complement of a $C_4$).
\end{enumerate}
 \end{enumerate}
\end{Def}

There are 3 more classes that will appear for technical reasons  in our results which are variants of the classes above.

\begin{Def}
\begin{enumerate}
\item A \textbf{\emph{2-star}} is a connected caterpillar with a dominating path of length at most 2, plus possibly isolated nodes.
\item A \textbf{\emph{1-split}} is a graph which is either a clique or a clique minus one edge, or the complement of such a graph, that is an independent set, plus possibly an edge.
\item An \textbf{\emph{augmented clique}} is a clique, plus one additional node with arbitrary adjacency, plus possibly isolated nodes. 
\end{enumerate}
\end{Def}

\paragraph*{Connectivity issues.}
A subtlety in the definitions of the classes is the connectivity of the graphs. For our work, there are basically two cases:
\begin{itemize}
\item The classes that are stable by disjoint union: linear forests, interval graphs, bipartite graphs, chordal graphs, comparability graphs, triangle-free graphs, permutation graphs, proper interval, caterpillars, trivially perfect
\item The classes that are stable only by addition of isolated vertices: stars, split graphs, threshold graphs, bipartite chain graphs, 2-stars, augmented clique.
\end{itemize}

Note that this classification depends on the class we chose to consider. That is, in this work we always consider both a class and its complement class, and for the complement classes the classification above becomes : stable by the join operation and by addition of a universal clique. 
Also in some cases, the class considered does not allow for isolated nodes, this is the case for 1-split for example. 

\paragraph*{Trivial classes.}
We used the word \emph{trivial} for classes that are basically finite. Here is a more formal definition. 

\begin{Def}
A graph class $\cal G$  is called \emph{trivial} if there exists a finite family of connected graphs $F$ such that $\forall G \in \cal G$ every connected component of $G$ is isomorphic to some graph in~$F$.
\end{Def}

We will use the following fact about these classes. 

\begin{Fa}\label{trivialclasses}
Let $\cal C$ be a trivial graph class and $P$ a pattern. The subclass of $\cal C$ consisting of the graphs that have a vertex ordering avoiding $P$, is also a trivial graph class.
\end{Fa}

\section{Characterization theorem}
\label{sec:main-theorems}

We are now ready to state our main theorem. We refer to \emph{basic operations} for addition of isolated nodes, and restriction to connected component.

\begin{Theo}\label{thm:characterization}
Up to complement and basic operations, the non-trivial classes that can be characterized by a set of patterns on three vertices are the following.
\end{Theo}

\begin{multicols}{3}
\begin{enumerate}
\item forests
\item linear forests
\item stars
\item interval 
\item split 
\item bipartite 
\item chordal
\item comparability 
\item triangle-free  
\item permutation 
\item threshold
\item proper interval
\item caterpillar
\item trivially perfect 
\item bipartite chain 
\item 2-star
\item 1-split 
\item augmented clique

\item bipartite permutation

\item triangle-free \\$\cap$ co-chordal
\item clique

\item complete bipartite
\end{enumerate}
\end{multicols}

Let us first say a few words about the proof, and then comment on this theorem.

\paragraph*{Proof outline.} To prove the theorem, we first generate all the split-minimal pattern families. This is basically done by listing all the possible sets of patterns, and then simplifying it, with the help of a program (that is presented in Subsection~\ref{subsec:program}).
The simplification step consists of removing the complement, mirror, complement-mirror families, and applying the pattern split rule until we get split-minimal families. 
After this step, we get the following result.

\begin{Lem}\label{lem:87-classes}
Up to complementation and mirroring, there exist 87 split-minimal families of patterns on three vertices.
\end{Lem}




The proof of the theorem consists in a case analysis: for each set of patterns of the list given by Lemma~\ref{lem:87-classes}, we find the class that is characterized by this set.
Roughly our approach is in three steps. 
First, we check whether the set of patterns is known to characterize a class. This works for the patterns of Theorem~\ref{thm:Damaschke} for example. 
If not, the second step is to use the structural properties of Section~\ref{sec:definitions} along with the class already characterized to get a candidate class. Sometimes this is enough, and we can conclude. Third, we work on the set of patterns, to get the class. Interestingly, for this last step, knowing various characterizations of a same class, as presented in Subsection~\ref{subsec:graph-classes}, does help.

\paragraph*{Comments on Theorem~\ref{thm:characterization}} 
The first comment to make is that there are only 22 classes in the list. 
The naive upper bound on the number of classes is $2^{3^3}=2^{27}$, thus having only 22 non-trivial classes (up to complement and basic operations) at the end is a surprising outcome. 
This is first explained by the fact that many sets of patterns are equivalent (because of the mirror, complement, and pattern split rules) as witnessed by Lemma~\ref{lem:87-classes}. (As will be explained later, when looking for the number of classes only, one can restrict to a set of eight patterns, thus the upper bound already falls to $2^8=256$, and further refinements makes the list go down to $87$.)
Also many classes end up being trivial because they have patterns that are somehow conflicting. 
For example Theorem~\ref{thm:complement} will show that in most cases, having a pattern and its complement in a family directly lead to a trivial class.
Finally, some classes, like threshold graphs appear several times in the proof, as they correspond to various sets of patterns. 

A second surprising fact is the large majority of the classes listed are well-known classes. This supports the insight of Theorem~\ref{thm:Damaschke}, that characterizations by patterns arise naturally in the study of graph classes.

Third, we phrased the theorem as a list, but the result is somehow richer. 
Indeed these classes form an interesting web of inclusions, and there is a lot to say about the relations between the classes.
These inclusions are represented in Figure~\ref{fig:diagram_3_patterns}.
(The inclusions that are not known or do not follow directly from Property~\ref{prop:basic}, are proved within the proof of the theorem.) 

\begin{figure}[!h]
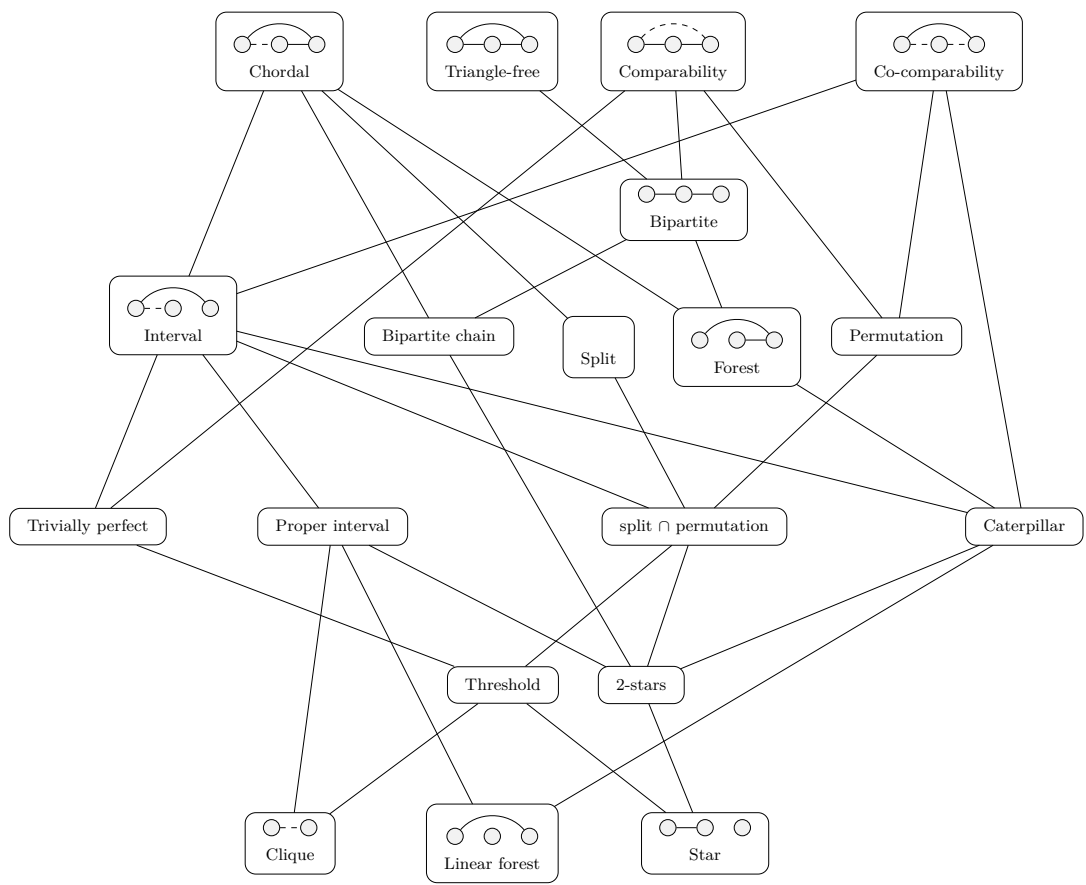

	\begin{center}
	\begin{minipage}{14.5cm}
		\scalebox{0.7}{

\begin{tikzpicture}   
[scale=2.0,auto=left,every node/.style={draw}]  

\clip (0.3,-1.4) rectangle (10.7,7.1);
 

\node[rectangle with rounded corners,rectangle corner radius=0.2cm]%
(caterpillar) at (10,2) {
	\begin{tabular}{c}
		Caterpillar
	\end{tabular}	
	};

\node[rectangle with rounded corners,rectangle corner radius=0.2cm]%
(bipartite-chain) at (4.5,3.8) {
	\begin{tabular}{c}
	Bipartite chain	
	\end{tabular}
};

\node[rectangle, rounded corners=0.2cm]%
(threshold) at (5.1,0.5){
	\begin{tabular}{c}
	Threshold
	\end{tabular}
};

\node[rectangle with rounded corners,rectangle corner radius=0.2cm]%
(twostars) at (6.4,0.5){
	\begin{tabular}{c}
	2-stars
	\end{tabular}
};

\node[rectangle with rounded corners,rectangle corner radius=0.2cm]%
(path) at (5.0,-1) 
	{\begin{tabular}{c}
		\begin{tikzpicture}[scale=0.7,auto=left, every node/.style=	{circle, draw, fill=black!5}]
	\node (a) at (0,0) {};
	\node (b) at (1,0) {};
	\node (c) at (2,0) {};
	\draw (a) to[bend left=50] (c);
\end{tikzpicture}\\
		Linear forest\\
	\end{tabular}
	}; 
	  
\node[rectangle with rounded corners,rectangle corner radius=0.2cm]%
(truc) at (6.9,2) 
	{\begin{tabular}{c}
		split $\cap$ permutation
	\end{tabular}
	};
	
\node[rectangle with rounded corners,rectangle corner radius=0.2cm]%
(star) at (7,-1) 
	{\begin{tabular}{c}
		\begin{tikzpicture}[scale=0.7,auto=left,every node/.style=	 {circle, draw, fill=black!5}]
	\node (a) at (0,0) {};
	\node (b) at (1,0) {};
	\node (c) at (2,0) {};
	\draw (a) to (b);

\end{tikzpicture}\\
		Star\\
	\end{tabular}
	};
	      
\node[rectangle with rounded corners,rectangle corner radius=0.2cm]%
(split) at (6.0,3.7)
  	{\begin{tabular}{c}
		\input{split.tex}\\
		Split\\
	\end{tabular}
	};    
	    
\node[rectangle with rounded corners,rectangle corner radius=0.2cm]%
(interval) at (2,4)
  	{\begin{tabular}{c}
		\begin{tikzpicture}
			[scale=0.7,auto=left,every node/.style=		
			{circle,draw,fill=black!5}]
			\node (a) at (0,0) {};
			\node (b) at (1,0) {};
			\node (c) at (2,0) {};
			\draw (a) to[bend left=50] (c);
			\draw[dashed] (a) to (b);
		\end{tikzpicture}\\
		Interval\\
	\end{tabular}
	};  
	
\node[rectangle with rounded corners,rectangle corner radius=0.2cm] (propinterval) at (3.5,2)
  	{	\begin{tabular}{c}
  		Proper interval\\
  	\end{tabular}
		
	};      

\node[rectangle with rounded corners,rectangle corner radius=0.2cm]%
(tree) at (7.3,3.7) 
	{\begin{tabular}{c}
		\begin{tikzpicture}
			[scale=0.7,auto=left,every node/.style=		
			{circle,draw,fill=black!5}]
			\node (a) at (0,0) {};
			\node (b) at (1,0) {};
			\node (c) at (2,0) {};
			\draw (a) to[bend left=50] (c);
			\draw (b) to (c);
		\end{tikzpicture}\\
		Forest\\
	\end{tabular}
	};        
	
\node[rectangle with rounded corners,rectangle corner radius=0.2cm]%
(bipartite) at (6.8,5)
	{\begin{tabular}{c}
		\begin{tikzpicture}
	[scale=0.7,auto=left,every node/.style=		
	{circle,draw,fill=black!5}]
	\node (a) at (0,0) {};
	\node (b) at (1,0) {};
	\node (c) at (2,0) {};
	\draw (a) to (b);
	\draw (b) to (c);
\end{tikzpicture}\\
		Bipartite\\
	\end{tabular}
	}; 
	
\node[rectangle with rounded corners,rectangle corner radius=0.2cm] (permutation) at (8.8,3.8)
	{\begin{tabular}{c}
		Permutation\\
	\end{tabular}
	};        
	
\node[rectangle, rounded corners=0.2cm] (trivially-perfect) at (1.2,2)
	{\begin{tabular}{c}
		Trivially perfect\\
	\end{tabular}
	};   
	     
\node[rectangle with rounded corners,rectangle corner radius=0.2cm] (chordal) at (3,6.5) 
	{\begin{tabular}{c}
		\begin{tikzpicture}
	[scale=0.7,auto=left,every node/.style=		
	{circle,draw,fill=black!5}]
	\node (a) at (0,0) {};
	\node (b) at (1,0) {};
	\node (c) at (2,0) {};
	\draw (a) to[bend left=50] (c);
	\draw[dashed] (a) to (b);
	\draw (b) to (c);
\end{tikzpicture}\\
		Chordal\\
	\end{tabular}
	};	    
		   
\node[rectangle with rounded corners,rectangle corner radius=0.2cm] (trianglefree) at (5,6.5)
 	{\begin{tabular}{c}
		\begin{tikzpicture}
			[scale=0.7,auto=left,every node/.style=		
			{circle,draw,fill=black!5}]
			\node (a) at (0,0) {};
			\node (b) at (1,0) {};
			\node (c) at (2,0) {};
			\draw (a) to[bend left=50] (c);
			\draw (a) to (b);
			\draw (b) to (c);
		\end{tikzpicture}\\
		Triangle-free\\
	\end{tabular}
	};      
	 
\node[rectangle with rounded corners,rectangle corner radius=0.2cm] (comparability) at (6.7,6.5)
  	{\begin{tabular}{c}
		\begin{tikzpicture}
			[scale=0.7,auto=left,every node/.style=		
			{circle,draw,fill=black!5}]
			\node (a) at (0,0) {};
			\node (b) at (1,0) {};
			\node (c) at (2,0) {};
			\draw[dashed] (a) to[bend left=50] (c);
			\draw (a) to (b);
			\draw (b) to (c);
		\end{tikzpicture}\\
		Comparability\\
	\end{tabular}
	};
	      
\node[rectangle with rounded corners,rectangle corner radius=0.2cm] (cocomparability) at (9.2,6.5)
  	{\begin{tabular}{c}	\input{cocomparability.tex}\\
		Co-comparability\\
	\end{tabular}
	}; 
	
\node[rectangle with rounded corners,rectangle corner radius=0.2cm] (clique) at (3.1,-1)
  	{\begin{tabular}{c}	\input{clique.tex}\\
		Clique\\
	\end{tabular}
	};

\draw (cocomparability) to (permutation);
\draw (comparability) to (permutation);

\draw (trianglefree) to (bipartite);
\draw (comparability) to (bipartite);

\draw (chordal) to (tree);

\draw (chordal) to (interval);
\draw (cocomparability) to (interval);

\draw (chordal) to (split);


\draw (interval) to (truc);

\draw (permutation) to (truc);
\draw (split) to (truc);

\draw (trivially-perfect) to (threshold);

\draw (chordal) to (bipartite-chain) ;
\draw (bipartite) to (bipartite-chain) ;

\draw (tree) to (caterpillar) ;
\draw (cocomparability) to (caterpillar) ;


\draw (interval) to (propinterval);
\draw (propinterval) to (path);
\draw (propinterval) to (twostars);

\draw (caterpillar) to (twostars) ; 
\draw (bipartite) to (tree) ;
\draw (twostars) to (star) ;
\draw (caterpillar) to (path) ;
\draw (threshold) to (star) ;
\draw (truc) to (twostars) ;	
\draw (bipartite-chain) to (twostars) ;
\draw (interval) to (caterpillar) ;
\draw (truc) to (threshold);
\draw (interval) to (trivially-perfect);
\draw (comparability) to (trivially-perfect);

\draw (threshold) to (clique);
\draw (propinterval) to (clique);

\end{tikzpicture}
 		}
 	\end{minipage}
 	\end{center}	 	
\caption{\label{fig:diagram_3_patterns}
Partial inclusion diagram of the classes that appear in Theorem~\ref{thm:characterization}. More refined diagrams can be found in Figures~\ref{fig:diagram-1-pattern}, \ref{fig:diagram-co-mirror}.}
\end{figure}

\begin{figure}
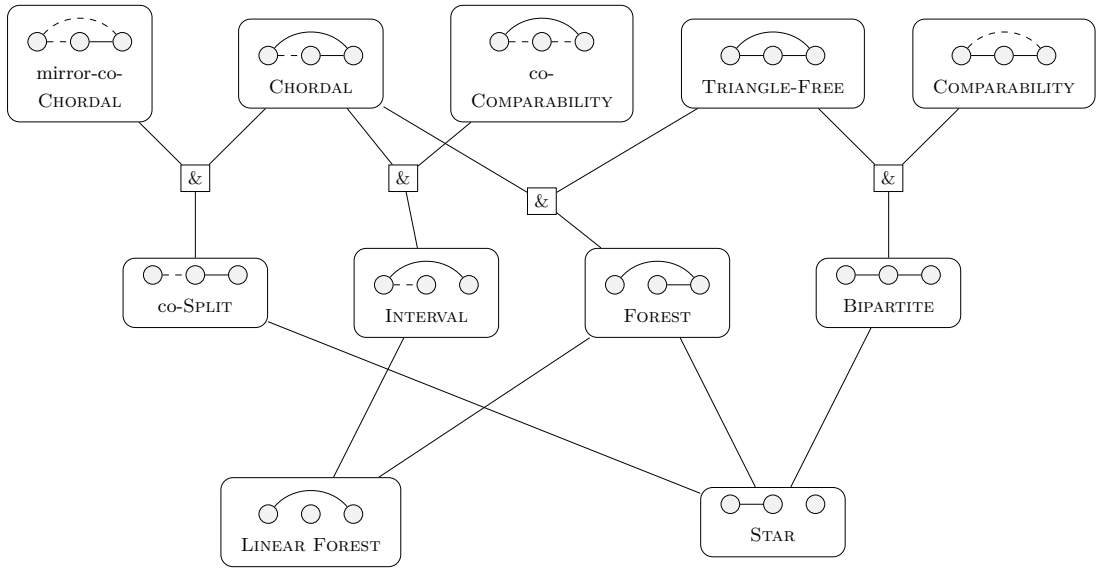

\begin{center}
\scalebox{0.8}{
\begin{tikzpicture}   
[scale=1.9,auto=left,every node/.style={draw}]  


\node[rectangle, rounded corners=0.2cm](mi-co-chordal) at (0,4) 
{\begin{tabular}{c}
	\input{reverse_cochordal.tex}\\
		mirror-co-\\
		\pchordal{} \\
	\end{tabular}
	}; 
	
\node[rectangle, rounded corners=0.2cm](chordal) at (2,4) 
{\begin{tabular}{c}
	\begin{tikzpicture}
	[scale=0.7,auto=left,every node/.style=		
	{circle,draw,fill=black!5}]
	\node (a) at (0,0) {};
	\node (b) at (1,0) {};
	\node (c) at (2,0) {};
	\draw (a) to[bend left=50] (c);
	\draw[dashed] (a) to (b);
	\draw (b) to (c);
\end{tikzpicture}\\
		\pchordal{} \\
	\end{tabular}
	}; 

\node[rectangle, rounded corners=0.2cm](co-comparability) at (4,4) 
{\begin{tabular}{c}
	\input{cocomparability.tex}\\
		co-\\
		\pcomparability{} \\
	\end{tabular}
	}; 

\node[rectangle, rounded corners=0.2cm](triangle-free) at (6,4) 
{\begin{tabular}{c}
	\begin{tikzpicture}
			[scale=0.7,auto=left,every node/.style=		
			{circle,draw,fill=black!5}]
			\node (a) at (0,0) {};
			\node (b) at (1,0) {};
			\node (c) at (2,0) {};
			\draw (a) to[bend left=50] (c);
			\draw (a) to (b);
			\draw (b) to (c);
		\end{tikzpicture}\\
		\ptriangle{} \\
	\end{tabular}
	}; 
	
	\node[rectangle, rounded corners=0.2cm](comparability) at (8,4) 
{\begin{tabular}{c}
	\begin{tikzpicture}
			[scale=0.7,auto=left,every node/.style=		
			{circle,draw,fill=black!5}]
			\node (a) at (0,0) {};
			\node (b) at (1,0) {};
			\node (c) at (2,0) {};
			\draw[dashed] (a) to[bend left=50] (c);
			\draw (a) to (b);
			\draw (b) to (c);
		\end{tikzpicture}\\
		\pcomparability{} \\
	\end{tabular}
	}; 

\node[rectangle, rounded corners=0.2cm](cosplit) at (1,2) 
{\begin{tabular}{c}
	\input{cosplit.tex}\\
		co-\psplit{} \\
	\end{tabular}
	}; 
	
\node[rectangle, rounded corners=0.2cm](interval) at (3,2) 
{\begin{tabular}{c}
	\begin{tikzpicture}
			[scale=0.7,auto=left,every node/.style=		
			{circle,draw,fill=black!5}]
			\node (a) at (0,0) {};
			\node (b) at (1,0) {};
			\node (c) at (2,0) {};
			\draw (a) to[bend left=50] (c);
			\draw[dashed] (a) to (b);
		\end{tikzpicture}\\
		\pinterval{} \\
	\end{tabular}
	}; 
	
\node[rectangle, rounded corners=0.2cm](forest) at (5,2) 
{\begin{tabular}{c}
	\begin{tikzpicture}
			[scale=0.7,auto=left,every node/.style=		
			{circle,draw,fill=black!5}]
			\node (a) at (0,0) {};
			\node (b) at (1,0) {};
			\node (c) at (2,0) {};
			\draw (a) to[bend left=50] (c);
			\draw (b) to (c);
		\end{tikzpicture}\\
		\pforest{} \\
	\end{tabular}
	};
	
\node[rectangle, rounded corners=0.2cm](bipartite) at (7,2) 
{\begin{tabular}{c}
	\begin{tikzpicture}
	[scale=0.7,auto=left,every node/.style=		
	{circle,draw,fill=black!5}]
	\node (a) at (0,0) {};
	\node (b) at (1,0) {};
	\node (c) at (2,0) {};
	\draw (a) to (b);
	\draw (b) to (c);
\end{tikzpicture}\\
		\pbipartite{} \\
	\end{tabular}
	}; 
	
\node[rectangle, rounded corners=0.2cm](path) at (2,0) 
{\begin{tabular}{c}
	\begin{tikzpicture}[scale=0.7,auto=left, every node/.style=	{circle, draw, fill=black!5}]
	\node (a) at (0,0) {};
	\node (b) at (1,0) {};
	\node (c) at (2,0) {};
	\draw (a) to[bend left=50] (c);
\end{tikzpicture}\\
		\ppath{} \\
	\end{tabular}
	}; 
	
\node[rectangle, rounded corners=0.2cm](star) at (6,0) 
{\begin{tabular}{c}
	\begin{tikzpicture}[scale=0.7,auto=left,every node/.style=	 {circle, draw, fill=black!5}]
	\node (a) at (0,0) {};
	\node (b) at (1,0) {};
	\node (c) at (2,0) {};
	\draw (a) to (b);

\end{tikzpicture}\\
		\pstar{} \\
	\end{tabular}
	};

\node (u-split) at (1,3) {\footnotesize{\&}};
\draw (mi-co-chordal) to (u-split);
\draw (chordal) to (u-split);
\draw (u-split) to (cosplit);

\node (u-interval) at (2.8,3) {\footnotesize{\&}};
\draw (chordal) to (u-interval);
\draw (co-comparability) to (u-interval);
\draw (u-interval) to (interval);

\node (u-forest) at (4,2.8) {\footnotesize{\&}};
\draw (chordal) to (u-forest);
\draw (triangle-free) to (u-forest);
\draw (u-forest) to (forest);

\node (u-bipartite) at (7,3) {\footnotesize{\&}};
\draw (triangle-free) to (u-bipartite);
\draw (comparability) to (u-bipartite);
\draw (u-bipartite) to (bipartite);

\draw (cosplit) to (star);
\draw (forest) to (star);
\draw (bipartite) to (star);
\draw (interval) to (path);
\draw (forest) to (path);

\end{tikzpicture}}
\end{center}
\caption{\label{fig:diagram-1-pattern} Refinement of Figure \ref{fig:diagram_3_patterns} in which we represent the cases where  $P=P_1\& P_2$ and the union-intersection property holds by a label \& link to $P_1$ and $P_2$ above, and $P$ below.}
\end{figure}

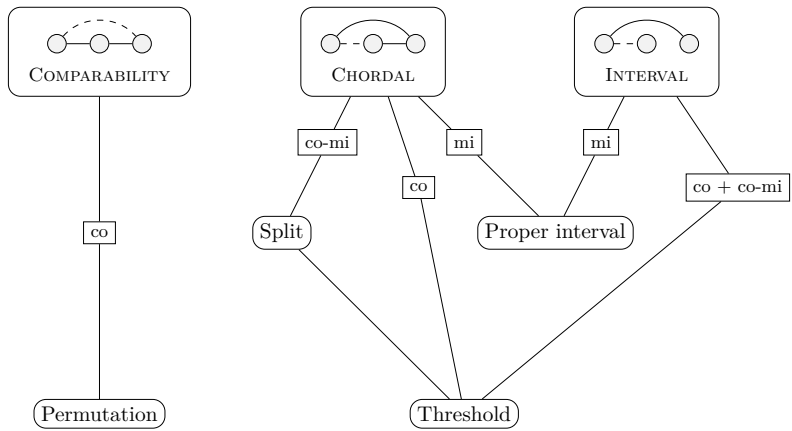
\begin{figure}
\begin{center}
\scalebox{0.8}{
{}\begin{tikzpicture}   
[scale=1.5,auto=left,every node/.style={draw}]  


\node[rectangle, rounded corners=0.2cm](comparability) at (0,4) 
{\begin{tabular}{c}
		\begin{tikzpicture}
			[scale=0.7,auto=left,every node/.style=		
			{circle,draw,fill=black!5}]
			\node (a) at (0,0) {};
			\node (b) at (1,0) {};
			\node (c) at (2,0) {};
			\draw[dashed] (a) to[bend left=50] (c);
			\draw (a) to (b);
			\draw (b) to (c);
		\end{tikzpicture}\\
		\pcomparability{} \\
	\end{tabular}
	}; 
	
\node[rectangle, rounded corners=0.2cm](chordal) at (3,4) 
{\begin{tabular}{c}
		\begin{tikzpicture}
	[scale=0.7,auto=left,every node/.style=		
	{circle,draw,fill=black!5}]
	\node (a) at (0,0) {};
	\node (b) at (1,0) {};
	\node (c) at (2,0) {};
	\draw (a) to[bend left=50] (c);
	\draw[dashed] (a) to (b);
	\draw (b) to (c);
\end{tikzpicture}\\
		\pchordal{} \\
	\end{tabular}
	};

\node[rectangle, rounded corners=0.2cm](interval) at (6,4) 
{\begin{tabular}{c}
		\begin{tikzpicture}
			[scale=0.7,auto=left,every node/.style=		
			{circle,draw,fill=black!5}]
			\node (a) at (0,0) {};
			\node (b) at (1,0) {};
			\node (c) at (2,0) {};
			\draw (a) to[bend left=50] (c);
			\draw[dashed] (a) to (b);
		\end{tikzpicture}\\
		\pinterval{} \\
	\end{tabular}
	};
	
\node[rectangle, rounded corners=0.2cm](split) at (2,2) {Split};

\node[rectangle, rounded corners=0.2cm](proper) at (5,2) {Proper interval};

\node[rectangle, rounded corners=0.2cm](threshold) at (4,0) {Threshold};

\node[rectangle, rounded corners=0.2cm](permutation) at (0,0) {Permutation};

\node (co-mi-chordal) at (2.5,3) {\footnotesize{co-mi}};
\draw (chordal) to (co-mi-chordal);
\draw (co-mi-chordal) to (split);

\node (co-chordal) at (3.5,2.5) {\footnotesize{co}};
\draw (chordal) to (co-chordal);
\draw (co-chordal) to (threshold);

\node (mi-chordal) at (4,3) {\footnotesize{mi}};
\draw (chordal) to (mi-chordal);
\draw (mi-chordal) to (proper);

\node (mi-interval) at (5.5,3) {\footnotesize{mi}};
\draw (interval) to (mi-interval);
\draw (mi-interval) to (proper);

\node (co-interval) at (7,2.5) {\footnotesize{co + co-mi}};
\draw (interval) to (co-interval);
\draw (co-interval) to (threshold);

\draw (split) to (threshold);

\node (co-comparability) at (0,2) {\footnotesize{co}};
\draw (comparability) to (co-comparability);
\draw (co-comparability) to (permutation);


\end{tikzpicture}}
\end{center}
\caption{\label{fig:diagram-co-mirror} Representation of the non-trivial characterizations of Theorems~\ref{thm:complement} and~\ref{thm:mirror}. The edges labeled with ``co'' mean: the family made by taking the pattern on the top endpoint and its complement characterize the class below. The edges labeled with ``mi'' mean the same but with mirror instead of complement. And ``co-mi'' designate both operations. The + means that the edge has both labels.}
\end{figure}


\section{Highlighted classes and characterizations}
\label{sec:special}

Among the sets of patterns on three nodes, some stand out because of their particular structures. 
These are, for example, the pairs of patterns that are mirror of one another. 
In this section, we study such special cases, anticipating on Theorem~\ref{thm:characterization}.
We start in Subsection~\ref{subsec:one-pattern}, with families restricted to one pattern and how they are related by the pattern split rule. 
Then, in Subsection~\ref{subsec:complementary}, we study pairs of complementary patterns, and in Subsection~\ref{subsec:mirror}, we study pairs of mirror (and mirror-complement) patterns. 

\subsection{Classes defined by one pattern}	
\label{subsec:one-pattern}

In this subsection, we consider the classes defined by one pattern. There are exactly $3^3=27$ different patterns on three nodes, as listed in Figure~\ref{fig:27-patterns}. 

We first prove Theorem~\ref{thm:Damaschke}.

\begin{proof}[Proof of Theorem~\ref{thm:Damaschke}]
For each class, when the characterization is not explicitly known, we first prove that it satisfies the forbidden pattern property by exhibiting the ordering, and then we show the other inclusion. Note that most of the graphs considered can be disconnected.
	
\textbf{\ppath{}.} ($\rightarrow$) If the graph is a linear forest that is a disjoint union of paths, then the natural ordering avoids the pattern. That is, placing the paths one after the other, from one endpoint to the other, avoids any jump over a node.

In the other direction, we first claim that if an ordering avoids the pattern, then every vertex has degree at most $2$. 
Indeed if a node has three or more neighbors, there must be at least two in the same direction (to the right or to the left), and then the pattern appears. Second, it is clearly not possible to avoid the pattern with a cycle. Thus we are left with the paths.

\textbf{\pstar{}.} In this paper, a star is a graph where at most one node has neighbors. This basically means that the nodes of a star can be partitioned into two sets $L$ and $I$ and a special node $c$, and the edge set is the union of the $(c,\ell)$ for any $\ell\in L$. 
($\rightarrow$) Any ordering with the nodes of $I$, then the nodes of $L$ and then $c$ avoids the pattern.
($\leftarrow$) Only the rightmost node can have an edge going left. Thus if we remove this node the graph is has no edges. This node is $c$ in our characterization, and its neighbors and non-neighbors define $L$ and $I$ respectively.

\textbf{\pinterval{}.}
This characterization is well-known for interval graphs, and an ordering avoiding this pattern is sometimes called a left-endpoint ordering, see \cite{RamalingamR88, Olariu91}. 
	
\textbf{\psplit{}.} Remember that a split graph is a graph that can be partitioned into an independent set and a clique. Note that there might be isolated nodes, belonging to the independent set.
($\rightarrow$) Any ordering of the following form avoid the pattern: first the isolated nodes, then the nodes from the independent set, and then the nodes from the clique.
($\leftarrow$)~Consider the node~$v$ that is the first node in the ordering to be the right-end of an edge, and let $u$ be the left end-point of such an edge.  
In order to avoid the pattern, the graph must contain all the edges~$(v,w)$ with $v<w$. Iteratively, we deduce that the graph contains all the edges $(a,b)$ with $v\leq a<b$. Therefore $v<...<n$ is a clique. Also by minimality, the nodes $1<...<v-1$ form an independent set. Hence the graph has a split partition. 
	
\textbf{\pforest{}.}
	($\rightarrow$) Consider the ordering $\tau$ given by any generic search applied on $G$, as defined in~\cite{CorneilK08}.
Suppose such $\tau$ contains the forbidden pattern on $a <_{\tau}b <_{\tau} c$ with $ac, bc \in E(G)$. Then using the four points condition of generic search, there must exist a vertex $d <_{\tau}b$, with $db \in E(G)$ and a path joining $a$ to $d$ with vertices before $d$ in $\tau$. This implies that there is a cycle in the graph, which is impossible as we started with a forest. Thus the pattern is not present.
	($\leftarrow$)~Consider a graph and any ordering $\tau$  of the vertices. Every cycle of the graph has a last vertex $x$ with respect to $\tau$. This vertex $x$ has necessarily two neighbours to its left in $\tau$, which corresponds to the forbidden pattern.
	
\textbf{\pbipartite{}.}   
	($\rightarrow$) Consider an ordering where the different connected components are placed one after the other. The pattern can only appear inside a component. Then for each component (that is bipartite) the vertices of one independent set are all placed before the vertices of the other independent set. 
All the vertices of the first set have all their edges pointing to the right, and all the vertices of the second set have all their edges pointing to the left. As a consequence no vertex has edges pointing both to the left and to the right, therefore the pattern does not appear. 
	($\leftarrow$) Consider an ordering of a graph avoiding the pattern. If the graph has no cycle, then it is bipartite. Consider now an induced cycle. Because of the forbidden pattern, the nodes of this cycle can be partitioned into two sets: the ones that are adjacent (in the cycle) to two nodes on their left, and the nodes that are adjacent (in the cycle) to two nodes on their right. Because any edge must have an endpoint in each set, the two sets must have the same size, and the cycle must have even length. Thus the graph is bipartite.\footnote{A generalization of this result appears in section \ref{sec:discussions}.}
	
\textbf{\pchordal{}.} This characterization is well-known, and the ordering is usually called \emph{simplicial elimination ordering} \cite{FulkersonG65}.\footnote{Note that it is perhaps more common to consider the reverse ordering, but this is equivalent, as we will see in Section~\ref{sec:definitions}.} 

\textbf{\pcomparability{}.} By definition an ordered graph avoids the forbidden pattern, if and only if, the ordering is a linear extension of a partial order. The fact that the complement pattern (sometimes called \emph{umbrella}) defines cocomparability graphs has been noted in \cite{KratschS93}.
		
\textbf{\ptriangle{}.} For patterns that are stable by any change of the ordering of the vertices, such as triangles, the forbidden pattern characterization boils down to the associated forbidden induced subgraph characterization. 

\textbf{At most two nodes.}
The pattern is made of three nodes, with no plain or dashed edge. Then every ordered graph with three or more nodes contains the pattern.
Thus the class is trivial: it consist only in the graphs with one or two nodes.
\end{proof}

The following corollary states that, up to complement, the classes defined by one pattern are exactly the ones listed in Theorem~\ref{thm:Damaschke}.

\begin{Coro}\label{coro:Damaschke-is-all}
Up to complementation the ten graph classes described in Theorem~\ref{thm:Damaschke}, are the only ones that can be defined with exactly one forbidden pattern on three nodes. 
\end{Coro}

\begin{proof}
We prove the statement in the following way: we count the number of different patterns obtained from the ones in Table~\ref{tab:one-pattern-3} by complementation and mirror, and check that we reach the total number, which is~27.
Hence we start with the 10 patterns of Table~\ref{tab:one-pattern-3}. All the patterns except the pattern \pnograph{}, have a complement that is not already in the list, thus we get to 19 patterns.
Now for the mirror, four patterns are self-mirrors (\pcomparability{}, \ptriangle{}, \pbipartite{} and \ppath{}), thus these do not add new patterns. The pattern \psplit{} has a mirror that is the same as its complement. Therefore we add only four mirror patterns, and get 23. Finally, these four patterns have complements that are not in the list yet. And this completes the landscape of 27 patterns. 

\end{proof}


Now, something interesting is the relations of these classes. It happens that in many cases the union-intersection property (as described in Subsection~\ref{subsec:union-intersection}) holds, which implies a neat hierarchy of inclusions represented in Figure~\ref{fig:diagram-1-pattern}.
The following lemma lists the cases of pattern split rule, and highlights the ones where the union-intersection property holds. 
In some cases we know it holds from the literature, and in some others, it is folklore.

\begin{Lem}\label{lem:1-split}
The following equalities hold:
\begin{enumerate}
\item \pforest{} = \pchordal{} \esp{} \ptriangle{},\\ and furthermore forests = chordal $\cap$ triangle-free.
\item \pbipartite{} = \pcomparability{} \esp{} \ptriangle{}{},\\ and furthermore  bipartite  = triangle-free $\cap$ comparability.
\item \psplit{} = \mirror\pchordal{} \esp{} \co\pchordal{},\\ and furthermore  split=chordal $\cap$ co-chordal  \cite{FoldesH77, HammerS81}.
\item \pinterval{} = \pchordal{} \esp{} \co\pcomparability{},\\ and furthermore  interval = chordal $\cap$ co-comparability \cite{GilmoreH64}.
\item \ppath{} = \pinterval{} \esp{} \mirror\pforest{} = \pforest{} \esp{} \mirror\pinterval{}.
\item \pstar{} = \pbipartite{} \esp{} \psplit{} = \mirror\pforest{} \esp{} \co\pinterval{}.
\end{enumerate}
\end{Lem}


For the four first items, the union-intersection property holds, and for completeness we prove that it does not hold for the last two items. For these items, the intersection of the classes is strictly larger than the class with the `unsplit' pattern.

\begin{Prop}\label{prop:not-union-intersection}
The following equalities hold.
\begin{enumerate}
\item The cycle-free interval graphs are the caterpillars. 
\item The bipartite split graphs are the 2-stars.
\item The cycle-free co-interval graphs are the 2-stars.
\end{enumerate}
\end{Prop}

\begin{proof}
We prove the three items.

\begin{enumerate}
\item From Item~\ref{item-def:caterpillar} of Definition~\ref{def:graph-classes}, caterpillars are the $(T_2, \text{cycle})$-free graphs. It is known that interval graphs also have a characterization by forbidden subgraphs \cite{LekkeikerkerB62}, and the only cycle-free subgraph in the list is $T_2$. Thus cycle-free interval graphs are exactly the caterpillars.
\item Consider a split graph and its partition into a clique $K$ and an independent set $I$. If the graph is bipartite, then $K$ has size at most two, as otherwise there would be a triangle. 
Then every node of the independent set can be connected to at most one of these clique nodes, for the same reason. This corresponds to a 2-star. The reverse inclusion is trivial, using the same partition. 
\item By definition a 2-star is cycle-free. We show that it is also a co-interval graph. 
Let $(a,b)$ be the dominating edge of the 2-star, $S_a$ be the leaves of $a$, and $S_b$ be the leaves of $b$. Also let $I$ be the set of isolated nodes.
Now we complement the graph. It has the following shape. The nodes of $I$ are connected to all the nodes, there is a complete bipartite graph between $S_a$ and $S_b$, $a$ is connected to every node of $S_b$, $b$ is connected to every node of $S_a$, and finally $S_a$ and $S_b$ are cliques. 
This can be represented by a set of closed intervals: $b$ is $[0,1]$, all the nodes of $S_a$ are $[1,2]$, all the nodes of $S_b$ are $[2,3]$, $a$ is $[3,4]$, and the nodes of $I$ are $[1,4]$. 
Thus a 2-star is a co-interval graph. 

Now we prove that cycle-free co-interval graphs are 2-stars. First note that a graph of the class cannot have two independent edges (that is a $2K_2$). Indeed the complement of two independent edges is a 4-cycle, and interval graphs are $C_4$-free. Thus, except for isolated nodes, the graph can have only one connected component. Also, as in a tree if the diameter is more than 3, then there are two independent edges. Therefore the connected component must be a tree of diameter at most 3, and this matches the definition of a 2-star.   
\end{enumerate}
\end{proof}

\subsection{Complementary patterns}
\label{subsec:complementary}

We now consider families of the type $\{P_1, P_2\}$, where $P_1$ and $P_2$ are complementary patterns.

\begin{Theo}\label{thm:complement}
The following characterizations hold:
\begin{enumerate}
\item \ptriangle{} $\cup$ \co\ptriangle{} defines  a trivial class.
\item \label{item:comp-co-comp}  \pcomparability{} $\cup$ \co\pcomparability{} defines the permutation graphs.
\item \pchordal{} $\cup$ \co\pchordal{} defines the threshold graphs.
\item \pinterval{} $\cup$ \co\pinterval{}  defines the threshold graphs.
\item \psplit{} $\cup$ \co\psplit{} defines the 1-split.
\item \pforest{} $\cup$ \co\pforest{}  defines a trivial class.
\item \pbipartite{} $\cup$ \co\pbipartite{} defines  a trivial class. 
\item \ppath{} $\cup$ \co\ppath{}  defines 
a trivial class.
\item \pstar{} $\cup$ \co\pstar{}  defines a trivial class.
\end{enumerate}
\end{Theo}

Diagram~\ref{fig:diagram-co-mirror} illustrates the non-trivial characterizations of Theorem~\ref{thm:complement} (and of Theorem~\ref{thm:mirror}).


\begin{proof}
\begin{enumerate}
\item
First notice that the patterns \ptriangle{} and co-\ptriangle{} are invariant by permutation of the vertices, thus excluding these patterns is a matter of induced subgraphs more than a matter of patterns.
It is known that any graph with at least six nodes has either a triangle or an independent set of size three (because the Ramsey number for these parameters is 6). Therefore only graphs with at most 5 vertices can belong to the class. The class is then trivial.

\item
It is known that the intersection of the classes of comparability and co-comparability graphs is the class of permutation graphs \cite{DushnikM41}. Thus the class of \pcomparability{} \& co-\pcomparability{} is included in the class of permutation graphs, by Item~\ref{item:union} of Property~\ref{prop:basic}. 
We use the literature to show that this inclusion is an equality. The 2-dimensional orders, which are transitive orientations of permutation graphs, were characterized in \cite{DushnikM41}, by the existence of a \emph{non-separating linear extension}. Such extension happens to be exactly an ordering avoiding the patterns \pcomparability{} and co-\pcomparability{}. 

\item \label{item:chordal-co-chordal} This result appears without proof in \cite{Damaschke90}. We provide a proof for completeness.

As written in Lemma~\ref{lem:1-split}, the intersection of chordal graphs and co-chordal graphs, is the class of split graphs. But the class of the union of the two patterns is smaller. 
Indeed, a split graph can have a $P_4$, and it is easy to check that no ordering of a $P_4$ can avoid the two patterns (see Figure~\ref{fig:P4}). 
The split graphs that have no $P_4$ are the threshold graphs (Item~\ref{item-def:threshold-P4} in Definition~\ref{def:graph-classes}). Thus the direct inclusion holds.


%

\begin{figure}[!h]
\begin{center}
\scalebox{1}{
\begin{tabular}{ccc}
\begin{tikzpicture}
	[scale=1,auto=left,every node/.style=		
	{circle,draw,fill=black!5}]
	\node (a) at (0,0) {};
	\node (b) at (1,0) {};
	\node (c) at (2,0) {};
	\node(d) at (3,0) {};
	\draw (a) to (b);
	\draw (b) to  (c);
	\draw (c) to (d);
	\draw[dashed] (a) to [bend left=50]  (d);
	\draw[dashed] (a) to [bend left=25](c);
\draw[dashed] (b) to [bend left=25](d);
\end{tikzpicture}

& 
\begin{tikzpicture}
	[scale=1,auto=left,every node/.style=		
	{circle,draw,fill=black!5}]
	\node (a) at (0,0) {};
	\node (b) at (1,0) {};
	\node (c) at (2,0) {};
	\node(d) at (3,0) {};
	\draw (a) to (b);
	\draw (b) to [bend left=50]  (d);
	\draw (c) to (d);
	\draw[dashed] (a) to [bend left=50]  (d);
	\draw[dashed] (a) to [bend left=25]  (c);
	\draw[dashed] (b) to  (c);
\end{tikzpicture}

&
\begin{tikzpicture}
	[scale=1,auto=left,every node/.style=		
	{circle,draw,fill=black!5}]
	\node (a) at (0,0) {};
	\node (b) at (1,0) {};
	\node (c) at (2,0) {};
	\node(d) at (3,0) {};
	\draw (a) to  [bend left=50](c);
	\draw (b) to  (c);
	\draw (b) to [bend left=50] (d);
	\draw[dashed] (a) to [bend left=50]  (d);
	\draw[dashed] (a) to (b);
	\draw[dashed] (c) to  (d);
\end{tikzpicture}\\

\begin{tikzpicture}
	[scale=1,auto=left,every node/.style=		
	{circle,draw,fill=black!5}]
	\node (a) at (0,0) {};
	\node (b) at (1,0) {};
	\node (c) at (2,0) {};
	\node(d) at (3,0) {};
	\draw (a) to  [bend left=50](c);
	\draw (c) to  (d);
	\draw (b) to [bend left=50] (d);
	\draw[dashed] (a) to [bend left=50]  (d);
	\draw[dashed] (a) to (b);
	\draw[dashed] (b) to  (c);
\end{tikzpicture}
&
\begin{tikzpicture}
	[scale=1,auto=left,every node/.style=		
	{circle,draw,fill=black!5}]
	\node (a) at (0,0) {};
	\node (b) at (1,0) {};
	\node (c) at (2,0) {};
	\node(d) at (3,0) {};
	\draw (a) to  [bend left=50](d);
	\draw (b) to  (c);
	\draw (b) to [bend left=50] (d);
	\draw[dashed] (a) to [bend left=25]  (c);
	\draw[dashed] (a) to (b);
	\draw[dashed] (c) to  (d);
\end{tikzpicture}

&
\begin{tikzpicture}
	[scale=1,auto=left,every node/.style=		
	{circle,draw,fill=black!5}]
	\node (a) at (0,0) {};
	\node (b) at (1,0) {};
	\node (c) at (2,0) {};
	\node(d) at (3,0) {};
	\draw (a) to  [bend left=50](d);
	\draw (b) to  (c);
	\draw (c) to (d);
	\draw[dashed] (a) to [bend left=25]  (c);
	\draw[dashed] (a) to (b);
	\draw[dashed] (b) to [bend left=25]  (d);
\end{tikzpicture}\\

\begin{tikzpicture}
	[scale=1,auto=left,every node/.style=		
	{circle,draw,fill=black!5}]
	\node (a) at (0,0) {};
	\node (b) at (1,0) {};
	\node (c) at (2,0) {};
	\node(d) at (3,0) {};
	\draw (a) to  [bend left=50](d);
	\draw (a) to  (b);
	\draw (c) to (d);
	\draw[dashed] (a) to [bend left=25]  (c);
	\draw[dashed] (b) to (c);
	\draw[dashed] (b) to [bend left=25]  (d);
\end{tikzpicture}

&

\begin{tikzpicture}
	[scale=1,auto=left,every node/.style=		
	{circle,draw,fill=black!5}]
	\node (a) at (0,0) {};
	\node (b) at (1,0) {};
	\node (c) at (2,0) {};
	\node(d) at (3,0) {};
	\draw (a) to  [bend left=50](d);
	\draw (a) to [bend left=25](c);
	\draw (b) to [bend left =50](d);
	\draw[dashed] (a) to  (b);
	\draw[dashed] (b) to (c);
	\draw[dashed] (c) to (d);

\end{tikzpicture}
&
\begin{tikzpicture}
	[scale=1,auto=left,every node/.style=		
	{circle,draw,fill=black!5}]
	\node (a) at (0,0) {};
	\node (b) at (1,0) {};
	\node (c) at (2,0) {};
	\node(d) at (3,0) {};
	\draw (a) to  (b);
	\draw (a) to [bend left=50](c);
	\draw (c) to (d);
	\draw[dashed] (a) to [bend left =50] (d);
	\draw[dashed] (b) to (c);
	\draw[dashed] (b) to [bend left=25] (d);

\end{tikzpicture}\\

\begin{tikzpicture}
	[scale=1,auto=left,every node/.style=		
	{circle,draw,fill=black!5}]
	\node (a) at (0,0) {};
	\node (b) at (1,0) {};
	\node (c) at (2,0) {};
	\node(d) at (3,0) {};
	\draw (a) to [bend left=50](d);
	\draw (a) to [bend left=25](c);
	\draw (b) to (c);
	\draw[dashed] (a) to  (b);
	\draw[dashed] (b) to [bend left=25] (d);
	\draw[dashed] (c) to (d);

\end{tikzpicture}
&
\begin{tikzpicture}
	[scale=1,auto=left,every node/.style=		
	{circle,draw,fill=black!5}]
	\node (a) at (0,0) {};
	\node (b) at (1,0) {};
	\node (c) at (2,0) {};
	\node(d) at (3,0) {};
	\draw (a) to  (b);
	\draw (a) to [bend left=50](c);
	\draw (b) to [bend left=25](d);
	\draw (a) to (b);
	\draw[dashed] (b) to  (c);
	\draw[dashed] (a) to [bend left=50] (d);
	\draw[dashed] (c) to (d);
\end{tikzpicture}
&
\begin{tikzpicture}
	[scale=1,auto=left,every node/.style=		
	{circle,draw,fill=black!5}]
	\node (a) at (0,0) {};
	\node (b) at (1,0) {};
	\node (c) at (2,0) {};
	\node(d) at (3,0) {};
	\draw (a) to  (b);
	\draw (a) to [bend left=50](d);
	\draw (a) to (b);
	\draw (b) to (c);
	\draw[dashed] (a) to [bend left=25] (c);
		\draw[dashed] (b) to [bend left=25] (d);
	\draw[dashed] (c) to (d);
\end{tikzpicture}
\end{tabular}
}
\end{center}
\caption{The 12 different orderings of a $P_4$.\label{fig:P4}}
\end{figure}
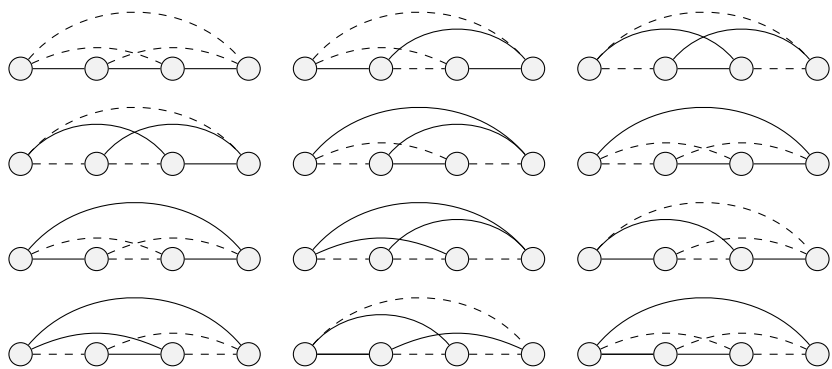

%

We now build an ordering of any threshold graph, that avoids both patterns. 
Note that an ordering $\tau$ avoids both patterns if and only if the reverse ordering of $\tau$ is a simplicial elimination ordering for both  $G$ and $\overline{G}$. 
We will build such an ordering, using the following claim.

\begin{Claim}\label{claim:threshold-x}
For any threshold graph there exists a partition of the vertices $V={x}\cup I \cup K$, such that $I$ is an independent set, $K$ is a clique, and $x$ is adjacent to every node in $K$ and no node in $I$.
\end{Claim}

\begin{proof}
A partition $V=K \cup I$ exists for any threshold graphs, as they are split graphs. It is sufficient to find a proper node $x$. Consider the node of $I$ that has the largest neighborhood inclusion-wise (it is well-defined, thanks to Item~\ref{item-def:threshold-inclusion} of Definition~\ref{def:graph-classes}). If this neighborhood is $K$, then this node can be taken as~$x$. Otherwise there exists a node of $K$ that has no neighbor in $I$ and it can be taken as~$x$.  
\end{proof}

The node $x$ of Claim~\ref{claim:threshold-x} is simplicial in both $G$ and $\overline{G}$. We can take it as the first vertex. Since threshold graphs form a hereditary class, we can repeat the argument on $G-x$. Therefore, every threshold graph admits a vertex ordering which is a simplicial elimination scheme of both $G, \overline{G}$. The reverse of this ordering avoids both patterns.
\item The class defined by \pinterval{} \& co-\pinterval{} is included in the class defined by \pchordal{} \& co-\pchordal{} because of Item~\ref{item:extension} of Property \ref{prop:basic}. Therefore this class is included in the threshold graphs.

For the other direction, we show that any threshold graph admits an ordering that avoids the patterns \pinterval{} and co-\pinterval{}. Using the pattern split rule, it is equivalent to find an ordering avoiding the patterns \pchordal{}, co-\pchordal{}, \pcomparability{} and co-\pcomparability{}. The ordering of the previous item avoids \pchordal{}, co-\pchordal{}. We show that it also avoids \pcomparability{} and co-\pcomparability{}. 
By complement, it is sufficient to show that the ordering avoids \pcomparability{}. 
Consider the reverse ordering where we added the nodes $x$, from left to right (this is harmless for the comparability pattern, as it is symmetric).
We remark that in this ordering, the nodes that are part of the independent set in the original partition, have no neighbor to their left.
Indeed when such a node is taken in the ordering, its full neighborhood is still present in the graph. 
Thus if the pattern appears, the first node must be part of the independent set, and the other nodes part of the clique. But this is impossible, as the nodes of the clique that are to the right of a node of the independent set must both be linked to this node.
\item
The only vertex orderings of a split graph that satisfy the \psplit{} pattern start with an independent set and finish with a clique. Analogously, for the complement pattern, we must have first a clique, and then an independent set. 
Consider the two first nodes of an ordering the avoids both patterns.  
Suppose they are linked by an edge. Then they cannot be both in an independent set, and then in the first partition, the independent set is reduced to one node. (We can always have at least one node in the independent set: if there is none, simply take a node from the clique to be in the independent set.)
Thus the rest of the nodes form a clique, and only the last node can be taken in the independent set of the second partition. 
As a consequence, all the pairs of nodes are linked, except possibly the first and last node. 
Thus the graph is a 1-split. 
By complement, if the first nodes are not linked then the graph is an independent set plus possibly one edge, which is also a 1-split. 
For the other direction, the ordering described above is sufficient. 
\item We show that the intersection of the classes of forests and co-forests is trivial (which is enough because of Item~\ref{item:union} in Property~\ref{prop:basic}).
A forest has at most $n-1$ edges, thus the union of the edges in the graph (which is a forest) and in the complement (which is also a forest) is at most $2n-2$. Thus $n(n-1)/2\leq 2n-2$, which holds only if $n\leq 4$. Thus the class is trivial.

\item Again we show that the intersection is trivial.
Consider a bipartite graph with more than two vertices is one of the parts. Its complement necessarily contains a triangle formed by nodes of this part. As a triangle prevents the complement from being bipartite, we know that no part of the original graph has more than two nodes. Thus no graph with more than four nodes belongs to the intersection.
\item Paths are special cases of forests, thus the intersection of paths and co-paths is trivial.
\item Stars are special cases of forests, thus the intersection of stars and co-stars is trivial.
\end{enumerate}
\end{proof}

A corollary extracted directly from the proof is the following.

\begin{Coro} \label{coro:order-threshold}
A graph $G$ is a threshold graph if and only if it admits a vertex ordering which is a simplicial elimination scheme for both $G$ and $\overline{G}$.
Furthermore this vertex ordering yields a transitive orientation of both $G$ and $\overline{G}$.
\end{Coro} 

\subsection{Mirror patterns}
\label{subsec:mirror}

We now study to the classes defined by a pattern and its mirror pattern, or its mirror-complement pattern. Note that the number of patterns to consider is rather small, because many patterns on three nodes are symmetric. 
Also note that the mirror of \psplit{} is co-\psplit{}, thus the associated class has already been considered in Theorem~\ref{thm:complement}.
Again see Diagram~\ref{fig:diagram-co-mirror} for illustration.  

\begin{Theo}\label{thm:mirror}
The following characterizations hold:
\begin{enumerate}
\item \pchordal{} $\cup$ \mirror\pchordal{} defines the proper interval graphs.
\item \pchordal{} $\cup$ \mirror\co\pchordal{}  defines the split graphs.
\item \label{item:interval-mirror-interval} \pinterval{} $\cup$ \mirror\pinterval{}  defines the proper interval graphs.
\item \pinterval{} $\cup$ \mirror\co\pinterval{} defines the threshold graphs.
\item \pstar{} $\cup$ \mirror\pstar{}  defines a trivial class.
\item \pstar{} $\cup$ \mirror\co\pstar{}  defines a trivial class.
\item \label{item:forest-mirror-forest} \pforest{} $\cup$  \mirror\pforest{}  defines paths.
\item \pforest{} $\cup$ \mirror\co\pforest{}  defines a trivial class.
\end{enumerate}
\end{Theo}

Note that proper interval graphs is an example of a class that can be defined via two different (split-minimal) pairs of forbidden patterns. 
Also, the patterns \pchordal{} and \pinterval{} illustrates the complex interaction that can exists between mirror, complement and mirror-complement patterns.

\begin{proof}
We prove the items one by one.
\begin{enumerate}
\item  An ordering that avoids both patterns is called \emph{reversible elimination scheme}, and it is known that the graphs that have such orderings are exactly the proper interval graphs~\cite{HabibL09}.
\item As already seen in Lemma~\ref{lem:1-split}, this class is the split graphs (by application of the Pattern split rule).
\item This characterization appears in \cite{LoogesO93}.
\item \label{item:interval-mirror-co-interval} The intersection of interval and co-interval graphs is included in the classes of split graphs (as a subclass of chordal and co-chordal graphs). 
Moreover, the graphs that avoid both patterns have no $P_4$ (see Figure~\ref{fig:P4}), thus, the class is included in threshold graphs. 
For the reverse inclusion, consider a threshold graph with the ordering $\tau=k_1, \dots k_q, i_p, \dots i_1$, using notations of Item~\ref{item-def:threshold-inclusion} in Definition~\ref{def:graph-classes}.
Suppose this ordering contains the pattern \pinterval{} on three nodes $a<b<c$. Necessarily $a\in K$ and $b,c\in I$ (using again the notation of the definition). 
Then by definition of $\tau$ we have $b=i_{\alpha}$ and $c=i_{\beta}$ with $\alpha>\beta$. Now because of the pattern  $N(i_\beta) \not\subseteq N(i_{\alpha})$ which is a contradiction. 
Similarly $\tau$ containing the pattern mirror-co-\pinterval{} draws a contradiction with the definition of the ordering.

\item  The only vertex orderings of a star graph that satisfies the \pstar{} pattern are forced to end with its center. For the mirror pattern we need to start from the center. Therefore only a star graph reduced to one edge can satisfy the two patterns.
\item The intersection of the classes defined by the patterns is trivial, as proved in Theorem~\ref{thm:complement}, thus the class is trivial.
\item Clearly the natural ordering for linear forests avoids these two patterns. Now suppose that a tree that is not a path avoids both patterns. This graph must have a claw ($K_{1, 3}$). But one can easily check that no ordering of a claw can avoid the two patterns. Therefore this class is reduced to paths.
\item The intersection of the classes is trivial,  as proved in Theorem~\ref{thm:complement}. 
\end{enumerate}

\end{proof}

\section{Proof of Theorem~\ref{thm:characterization}}
\label{sec:main-proof}


\subsection{Program generating the list of classes}
\label{subsec:program}

In this subsection, we describe and prove the algorithm used for generating all the pattern families of Theorem~\ref{thm:characterization}. Note that the task of listing the relevant pattern families could also be carried out by hand, but we think that writing and proving a program produces a more reliable output, reducing the risk of missing a family.

Let us first provide some notations and intuition to understand the algorithm that follows. 
First, we will encode the families of patterns by bit vector of 27 cells (following the numbering of the 27 patterns of Figure~\ref{fig:27-patterns}). (In the algorithm the numbering of the cells starts at 0.)
Actually we will start with 8 bits vectors, because we start with only full patterns. 
Second the algorithm basically consists in generating all the families of full patterns (line 1), then simplifying it by removing equivalent families (lines 2-4), and then we further simplify by using the pattern split rule (lines 6-8).
Third, we use the following notations: $complement(V)$ designates the fact of reversing the vector (the first bit becomes the last bit etc.), and $exchange(V)$ is the vector $V$ where the bits at position 1 and 4 have been exchanged, as well as the ones at position 3 and 6.   

\begin{algorithm}
\KwResult{A list of patterns}

$D \leftarrow$ all 8-bit vectors except [0,0,0,0,0,0,0,0], \\

\For{ all pairs $p, q \in D, p \neq q$}{
	\If{ $(p=complement(q))\vee(p=exchange(q))\vee (p=exchange(complement(q))$)}{Remove $q$ from $D$}}	
	
Append 19 zeros after each bit vector in $D$ \{\% to upgrade them as 27-bit vectors\},\\

\For{all $p\in D$}{
 \For{every triplet $(a,b,c)$ of Lists~$A, B, C$ in Figure~\ref{fig:triplets}}{ if the $a^{\text{th}}$ and the $b^{\text{th}}$ bit of $p$ are 1s  then transform them into 0s, and make the $c^{\text{th}}$ bit of $p$ a 1. }}
Return $D$.
\caption{Computing split-minimal patterns}	
\end{algorithm}

\begin{figure}[!h]
\begin{description}
\item{List A:} $(0,1,16)$, $(0,2,12)$, $(0,4,8)$, $(1,3,13)$, $(1,5,9)$, $(2,3,17)$, $(2,6,10)$, $(3,7,11)$, $(4,5,18)$, $(4,6,14)$, $(5,7,15)$, $(6,7,19)$.
\item{List B:} $(8,9,22)$, $(8,10,20)$, $(9,11,21)$, $(10,11,23)$, $(12,13,24)$, $(12,14,20)$, $(13,15,21)$, $(14,15,25)$, $(16,17,24)$, $(16,18,22)$, $(17,19,23)$, $(18,19,25)$.
\item List C: $(20,21,26)$, $(22,23,26)$, $(24,25,26)$.
\end{description}
\caption{\label{fig:triplets}Lists of triplets describing the pattern split rule.}
\end{figure}

We now prove the correctness of the algorithm.

\begin{claim}
All the classes characterized by patterns on three nodes are captured by a family described by a bit vector of $D$ at line 5 (up to complement).
\end{claim}

As noted at the end of Subsection~\ref{subsec:pattern-split}, all the families are equivalent to a family of full patterns. Thus after line~1, $D$ has the property of the claim. Note that we do not consider the trivial family that contains no patterns and corresponds to the all-zero vector.
For line 2 to 4, note first that the \emph{exchange} operation, is the equivalent on vectors of a mirror operation on patterns (because Pattern~1 (resp. 3) is the mirror of Pattern~4 (resp.~6) and that the other full patterns are stable by mirror). 
Then note that because of our numbering of the patterns in Figure~\ref{fig:27-patterns}, the fact of reversing a vector (on 8 bits) is equivalent of complementing each pattern of the family.
Then the removal of line~4 is safe: the only vectors that are removed are the ones that represent a family that has a complement, a mirror or a complement mirror equivalent whose vector stays in~$D$.

\begin{claim}
All the triplets $(a,b,c)$ listed in Figure~\ref{fig:triplets} are such that the pattern $c$ can be split into $a$ and $b$ by the pattern split rule. 
\end{claim}

This can be checked triplet by triplet. As stated in Lemma~\ref{lem:pattern-split}, the class defined by a family is stable by the split (or here the reverse split) of a pattern. Thus we still have captured all the families at the end of the algorithm. Also because we apply the simplification with patterns of 0, then 1 and then 2 undecided edges, we saturate the set with the simplification rule, and all the families are split minimal by the end of the algorithm.

\subsection{Enumeration}

\label{subsec:big-proof}

\paragraph*{Proof strategy and notations.}
The program described in Subsection~\ref{subsec:program}, provides a list of 87 families to investigate in order to have a list of all the classes characterized by a family of patterns on three nodes. 
The proof of Theorem~\ref{thm:characterization} consists in going through this list and finding the class associated with each family.
We use the following notation: [X]\textsubscript{\textcolor{gray}{$i$}} denotes the family composed of the patterns in the list $X$, described by their numbers from Figure~\ref{fig:27-patterns}; and $i$ is just an index to keep track of the items in the list of 87 families. 
For example [2,5]\textsubscript{\textcolor{gray}{2}} is the second family in our list, and it is composed of the patterns 2 and 5, which are respectively \pcomparability{} and co-\pcomparability{}.\\
Note that for one class, there may be several families, for example through mirror operation; in order to make the comparison with the output of the program easier, we do not change the family to make it ``more canonical''. For example, the program outputs [1] and not [4], which means that we consider the pattern mirror-\pchordal{} and not the pattern \pchordal{}.

\paragraph*{Isolated nodes and connectivity.}

In the majority of the classes, additional isolated nodes are allowed. In our orderings they often appear on the left of the other nodes. But some pattern forbid such isolated nodes, except when the graph is reduced to an independent set, in which case we indicate it in parenthesis. 
Also for some classes, some combination of patterns allow several non-trivial connected component, and some allow only one connected component. Again this is denoted in parenthesis.

\paragraph*{Enumeration.}
\begin{description}
\item[\family{2}{1}{Comparability graphs}] See Theorem~\ref{thm:Damaschke}.
\item[\family{2,5}{2}{Permutation graphs}] See Theorem~\ref{thm:complement}.
\item[\family{1}{3}{Chordal graphs}] See Theorem~\ref{thm:Damaschke}, in addition to the mirror property (Item~\ref{item:mirror} in Property~\ref{prop:basic}).
\item[\family{1,6}{4}{Threshold graphs}] The mirror property (Item~\ref{item:mirror} in Property~\ref{prop:basic}) implies that this family is equivalent to [4,3], that is  \pchordal{} and co-\pchordal{}. Then the result follows from Theorem~\ref{thm:complement}.
\item[\family{9}{5}{Interval graphs}] See Theorem~\ref{thm:Damaschke} (and mirror property).
\item[\family{1,4}{6}{Proper interval graphs}]
See Theorem~\ref{thm:mirror}.
\item[\family{4,9}{7}{Proper interval graphs}]
We can note that Pattern~1 extends Pattern~9, and Pattern~4 extends Pattern~18. Thus we know that the class defined by [4,9] is contained in the one of [1,4] and contains the one of [9,18]. 
From Theorem~\ref{thm:mirror}, we know that [1,4] and [9,18] both define proper interval graphs, thus the result. 
\item[\family{13}{8}{Split graphs}] See Theorem~\ref{thm:Damaschke}.
\item[\family{4,13}{9}{Augmented cliques}]

Let us first give an ordering that avoids the patterns: first the isolated nodes, then the special node, then the nodes of the clique that are neighbor to this node, and then the rest of the clique. It is easy to check that this ordering avoids both patterns. 

Now, for the reverse direction. Thanks to the proof of Theorem~\ref{thm:Damaschke} for the split graphs, we know that a graph that avoids Pattern~13 can be divided into first an independent set and then a clique. The pattern~4 forbids that a node of the clique is adjacent to two nodes of the independent set. Thus every node of the clique has at most one neighbor in the independent set, and we claim that this neighbor is the same for every such node. Suppose there is the following scenario in an ordered graph: two nodes $x$ and $y$ from the independent set, and then two nodes $u$ and $v$ from the clique such that $(x,u)$ and $(y,v)$ are edges. One can check that for any ordering of $x$ and $y$, Pattern~13 appears. This proves the claim. 
And this finishes the proof: the graph must be a clique, with possibly some nodes linked to one additional node (plus isolated vertices).

\item[\family{13,14}{10}{1-Split}]
See Theorem~\ref{thm:complement}.
\item[\family{1,2}{11}{Trivially perfect graphs}]
Thanks to Figures~\ref{fig:P4} and~\ref{fig:C4}, one can check that no ordering of $C_4$ or $P_4$ can avoid both patterns.
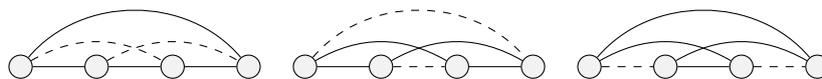
\begin{figure}[!h]
\begin{center}
\begin{tabular}{ccc}
\begin{tikzpicture}
	[scale=1,auto=left,every node/.style=		
	{circle,draw,fill=black!5}]
	\node (a) at (0,0) {};
	\node (b) at (1,0) {};
	\node (c) at (2,0) {};
	\node(d) at (3,0) {};
	\draw (a) to  (b);
	\draw (b) to (c);
	\draw (c) to (d);
	\draw[dashed] (a) to [bend left=30] (c);
	\draw[dashed] (b) to [bend left=30] (d);
	\draw (a) to [bend left=50] (d);
\end{tikzpicture}
&
\begin{tikzpicture}
	[scale=1,auto=left,every node/.style=		
	{circle,draw,fill=black!5}]
	\node (a) at (0,0) {};
	\node (b) at (1,0) {};
	\node (c) at (2,0) {};
	\node(d) at (3,0) {};
	\draw (a) to  (b);
	\draw[dashed] (b) to (c);
	\draw (c) to (d);
	\draw (a) to [bend left=30] (c);
	\draw (b) to [bend left=30] (d);
	\draw[dashed] (a) to [bend left=50] (d);
\end{tikzpicture}
&
\begin{tikzpicture}
	[scale=1,auto=left,every node/.style=		
	{circle,draw,fill=black!5}]
	\node (a) at (0,0) {};
	\node (b) at (1,0) {};
	\node (c) at (2,0) {};
	\node(d) at (3,0) {};
	\draw[dashed] (a) to  (b);
	\draw (b) to (c);
	\draw[dashed] (c) to (d);
	\draw (a) to [bend left=30] (c);
	\draw (b) to [bend left=30] (d);
	\draw (a) to [bend left=50] (d);
\end{tikzpicture}
\end{tabular}
\end{center}
\caption{\label{fig:C4} The three different orderings of a $C_4$.}
\end{figure}
Therefore, by Item~\ref{item:p4-c4} of Definition~\ref{def:graph-classes}, this class is included in the class of trivially perfect graphs.

For the other direction, we give two proofs. First, as said in Item~\ref{item:comparability-of-tree} of Definition~\ref{def:graph-classes}, a trivially perfect graph is the comparability graph of an arborescence. 
Let us orient each connected component of this arborescence from the leaves to the root. Consider an ordering given by the following rule: take a source delete it and update the graph. It yields a linear extension of the graph and avoids Pattern~2 (\pcomparability{}).
This ordering is also a perfect elimination scheme and avoids Pattern~1 (mirror-\pchordal{}).

Now a second proof is based on Item~\ref{item:quasi-threshold} in Definition~\ref{def:graph-classes}: the class of trivially perfect graph, is the class that contains the graph with one node, and is stable by disjoint union and addition of a universal vertex.
We build the ordering following this incremental construction: in the the disjoint union we put the different component one after the other, and we add the universal vertices on the right-most position in the ordering. It is easy to check that this construction avoids the patterns. 

\item[\family{1,10}{12}{Threshold graphs}]
Using the Pattern Split rule (Lemma~\ref{lem:pattern-split}), this family is equivalent to [1,2,6]. 
Syntactically, we deduce that it is included in the class of $[1,6]_4$, and this family is mirror of $\pchordal{}\cup \co\pchordal$, which defines threshold graphs (Theorem~\ref{thm:complement}). 
Also Pattern~1 extends Pattern~9, thus $[1,10]$ contains $[9,10]$, which is the mirror of $\pinterval\cup \co\pinterval$, which also defines threshold graphs (Theorem~\ref{thm:complement}). Thus [1,10] defines threshold graphs.

\item[\family{2,9}{13}{Trivially perfect graphs}]

As Pattern~1 extends Pattern~9, $[1,2]_{11}$ extends $[2,9]$, thus the class is included into trivially perfect graphs. Also, the ordering used for the proof  of $[1,2]_{11}$ avoids the Patterns of $[2,9]$ thus the result.

\item[\family{9,10}{14}{Threshold graphs}] See Theorem~\ref{thm:complement}-4 (along with the mirror property of Item~\ref{item:mirror} in Property~\ref{prop:basic}).
\item[\family{1,2,4}{15}{Cliques}]
We claim that no ordering of a $P_3$ avoids the three patterns. Indeed, if the middle node of the path is respectively first, second or third in the ordering, then the pattern respectively 1, 2 or 4, is present. 
Thus the class is included in the union of cliques, and any ordering where the cliques appear one after the other avoids the three patterns.
\item[\family{2,4,9}{16}{Cliques}]
This case is similar to the previous one, $[1,2,4]_{15}$, where 9 replaces 1. The same ordering works.
\item[\family{2,13}{17}{Threshold graphs}]

Because of Pattern~13, the class is included into the class of split graph, and using Figure~\ref{fig:P4}, one can check that no ordering of a $P_4$ avoids both patterns. Thus the class is included into threshold graphs, thanks to Item~\ref{item-def:threshold-P4} in Definition~\ref{def:graph-classes}.
For the other inclusion, using the notations of Item~\ref{item-def:threshold-inclusion} in Definition~\ref{def:graph-classes}, the ordering $i_1,...,i_p,k_q,...k_1$ avoids both patterns.  
\item[\family{10,13}{18}{Co-augmented-cliques}]

Note that the complement-mirror family of [10,13] is [18,13], which is extended by $[4,13]_9$, which characterize augmented cliques. Also note that ordering given in $[4,13]_9$ also avoids Pattern~18. 
Thus the result. 

\item[\family{2,5,13}{19}{Threshold graphs}]
This class is included into the one of $[2,13]_{17}$ thus it is included into threshold graphs. 
For the other inclusion, using again the notations of Item~\ref{item-def:threshold-inclusion} in Definition~\ref{def:graph-classes}, the ordering $i_1,...,i_p,k_q,...k_1$ avoids the three patterns.  
\item[\family{2,4,13}{20}{Connected cliques}]
Similarly to the case of $[1,2,4]_{15}$, the patterns forbid an induced $P_3$, thus every connected component is a clique. If there are two component with one or more edges, then Pattern~13 appears, thus the graphs of the class are composed of one clique, plus possibly isolated vertices. 
In the other direction, having the isolated vertices first, and then the clique avoids the pattern. 

\item[\family{4,10,13}{21}{Connected cliques (without isolated nodes or of size < 3)}]
As Pattern~2 extends Pattern~10, the class of [4,10,13] is included into the one of $[2,4,13]_{20}$, that is in the connected cliques. Moreover, if there is an isolated node, it cannot be on the left or on right of an edge because of respectively Pattern~10 and~13. As a consequence, isolated nodes can appear only if the clique has strictly less than three nodes.
An ordering avoiding the patterns has first the possible isolated nodes and then the clique.
\item[\family{2,13,18}{22}{Connected cliques}]
As Pattern~4 extends Pattern~18, the class of [2,13,18] is included into the one of $[2,4,13]_{20}$, that is in the connected cliques. And the same ordering works.
\item[\family{10,13,18}{23}{Connected cliques}]
As Pattern~2 extends Pattern~10, the class is included in the class of $[2,13,18]_{22}$,  that is in the connected cliques. And the same ordering works.

\item[\family{0}{24}{Triangle-free graphs}] See Theorem~\ref{thm:Damaschke}.
\item[\family{0,7}{25}{Trivial}]
See Theorem~\ref{thm:complement}
\item[\family{0,5}{26}{Bipartite permutation}]
First we show that the graphs of the class are bipartite. It is well known  that co-comparability graphs do not contain induced odd cycles of length $\geq 5$ as induced subgraphs~\cite{Gallai67}. 

Thus if there is an odd cycle, it is a triangle, and this is forbidden by Pattern~0. 
Then the graphs at hand are contained in bipartite cocomparability graphs. Moreover, bipartite graphs are comparability graphs, and comparability co-comparability graphs are permutation graphs. Therefore the class is contained in the class of bipartite permutation graphs. 
And any cocomparability ordering avoids both patterns, as there is no triangle.

Note that characterizations of bipartite permutation graphs related to orderings are known~\cite{SpinradBS87}, but they are different. 

\item[\family{0,3}{27}{Triangle-free $\cap$ co-chordal}]
by Fact~\ref{fact:order-invariant}.
\item[\family{0,3,6}{28}{Bipartite chain graph}]
Patterns~3 and~6 imply that  this class is contained into co-proper-interval graphs (thanks to the complement property and Item~\ref{item:interval-mirror-interval} of Theorem~\ref{thm:mirror}). 
In addition, it is known that co-interval graphs are comparability graphs of proper interval orders~\cite{Roberts69}.
If in this order, there are three elements $x<y<z$, then there would be a triangle in the comparability graph, which is forbidden by Pattern~0. 
Thus the graphs considered are bipartite. 
Furthermore co-interval graphs do not contain $2K_2$ (because intervals do not contain induced $C_4$\cite{LekkeikerkerB62}), thus these graphs are included into bipartite chain graphs (by Item~\ref{item:bip-chain-forbidden} of Definition~\ref{def:graph-classes}).

Now consider a bipartite chain graph, and the notations of Item~\ref{item:bip-chain-neighborhood} of Definition~\ref{def:graph-classes}.
The vertex ordering $a_1, \dots,  a_{|A|}, b_1, \dots,  b_{|B|}$  avoids the three patterns. 
\item[\family{0,3,5}{29}{Complete bipartite graphs}]

The class is contained into the class of $[0,5]_{26}$, that is in the class of bipartite permutation graphs. 
Also because of Pattern~3, we know that there is at most one connected component that is not an isolated node (otherwise the right-most vertex, and two vertices of another non-trivial component would make it appear).  
If the non-trivial component is not a complete bipartite graph, then the following situation should appear: two nodes $a,b$ on one side, two nodes $c,d$, on the other side, $(b,c)$ \emph{not} being an edge (because the bipartite graph is not complete), $(a,c)$ and $(b,d)$ being edges (to insure the connectivity), and $(a,d)$ being arbitrary. 
Now note that because of Patterns~3 and~5, if a node has two non-neighbors that are linked by an edge, then this node should appear before the two other nodes.  
This means that $b$ is before $a$ and $c$, and $c$ is before $b$ and $d$ which is a contradiction. Thus the completeness of the bipartite graph.

The ordering with isolated nodes and then one class after the other avoids the three patterns.
\item[\family{0,3,5,6}{30}{Complete bipartite graphs (without isolated nodes)}]
As a subclass of case $[0,3,5]_{29}$, it is included into complete bipartite graphs, and the patterns~3, 5 and~6, forbid isolated nodes respectively, on the right, middle and left of an edge. For the other direction, the same ordering as for Pattern~6 works (without the isolated nodes).

\item[\family{12}{31}{Bipartite graphs}] See Theorem~\ref{thm:Damaschke}.
\item[\family{7,12}{32}{Trivial class}]
co-\ptriangle{} \& \pbipartite{} defines a trivial class, since such a graph can have at most 4 vertices.
\item[\family{5,12}{33}{Bipartite permutation graphs}]
By the pattern split rule, this class is equivalent to the one of $[0,2,5]$. 
Remember  that $[2,5]_2$ characterizes permutation graphs (by Item~\ref{item:comp-co-comp} of Theorem~\ref{thm:complement}) and that Pattern~12 characterizes bipartite graphs (Theorem~\ref{thm:Damaschke}), thus the class is included into bipartite permutation graphs.
For the other inclusion, the ordering given for Item~\ref{item:comp-co-comp} of Theorem~\ref{thm:complement} avoids Pattern~2 and~5, and it also avoids Pattern~0, as the graph is bipartite.  
\item[\family{12,15}{34}{Trivial class}] See Theorem~\ref{thm:complement}.
\item[\family{3,12}{35}{Bipartite chain graphs}]
The class contains only bipartite graphs because of Pattern~12. A graph of the class cannot contain a $2K_2$, as the right-most node of a $2K_2$ in the ordering would form Pattern~3 with the nodes of the other edge. Thus the class is included in bipartite chain graph by Item~\ref{item:bip-chain-forbidden} of Definition~\ref{def:graph-classes}.

As for $[0,3,6]_{28}$, the vertex ordering with first the isolated nodes and then $a_1, \dots,  a_{|A|}$, $ b_1, \dots,  b_{|B|}$ (with the notations of Item~\ref{item:bip-chain-neighborhood} of Definition~\ref{def:graph-classes}) avoids the patterns .
\item[\family{3,6,12}{36}{Bipartite chain graphs}]

Because of case $[3,12]_{35}$, this class is included in bipartite chain graphs. Also the ordering $a_1, \dots,  a_{|A|}$, the isolated nodes and then $b_1, \dots,  b_{|B|}$ (with the same notation as above) works. 
\item[\family{3,5,12}{37}{Complete bipartite graphs}]
As Pattern~0 extends Pattern~12, the class is included into the class of $[0,3,5]_{29}$, that is in the complete bipartite graphs.
The ordering with isolated nodes and then one class after the other avoids the three patterns.
\item[\family{3,5,6,12}{38}{Bipartite chain graphs (without isolated nodes)}]
Again the inclusion follows from $[3,12]_{35}$, and from the fact that Patterns~3, 5 and~6 forbid isolated nodes. The ordering $a_1, \dots,  a_{|A|}, b_1, \dots,  b_{|B|}$ avoids the three patterns
\item[\family{16}{39}{Forest}] See Theorem~\ref{thm:Damaschke} (and mirror property). 
\item[\family{7,16}{40}{Trivial}]
This class is included into the class of forests with maximum independent set of size at most two. 
Such trees have degree at most 2 and diameter at most~4. 
There are only a finite number of such trees.
\item[\family{6,16}{41}{Stars}]
Pattern~16 implies that the graph of the class are forests. Also one can check on Figure~\ref{fig:P4} that no ordering of $P_4$ avoids both patterns, thus the maximum diameter of the trees of the forest is at most 2. Finally Pattern~6 implies that there cannot be several non-trivial components. 
Thus the class is included into the stars. 
Any ordering of a star with first the leaves, then the center, and then the isolated nodes avoids both patterns.
\item[\family{16,19}{42}{Trivial class}] See Theorem~\ref{thm:complement}, and the mirror property.
\item[\family{5,16}{43}{Caterpillars}]
This class is included into forests $\cap$ co-comparability. As any connected component of a cocomparability graphs admit a dominating path \cite{CorneilOS97}, the class is included in the caterpillars. 
For the other inclusion, first note that using the mirror property, the two families [5,16] and [5,8] define the same graph class. Then it is known that for any caterpillar, two BFS provide an ordering avoiding Patterns~5 and~8 \cite{CharbitHMN17}.	
\item[\family{15,16}{44}{Trivial}]
The class is included into co-bipartite forests.
We recall that a graph is co-bipartite if its the vertices can be partitioned into two cliques, with some edges in between. Since we have to intersect with forests, the two cliques have at most two vertices. Therefore each connected component has at most four vertices.
\item[\family{5,6,16}{45}{Stars}]
Using the family $[6,16]_{41}$, this class is included in the stars. And the vertex ordering used for $[6,16]_{41}$ also avoids Pattern~5.
\item[\family{4,16}{46}{Linear forests}]
From Pattern~16, this class is included into forests. 
Now suppose that the graph contains a star with 3 leaves and a center $x$. Pattern~16 forbids $x$ to have 2 leaves on its right and Pattern~4 forbids $x$ to have 2 leaves on its left. Since in any ordering $x$ has at least 2 leaves either on its left or its right, such stars with three leaves cannot appear in graphs in the class.
 Therefore each connected component is a path. Hence this class corresponds to linear forests
\item[\family{4,7,16}{47}{Trivial}]
As a subclass of $[7,16]_{40}$.
\item[\family{22}{48}{Linear forests}] See Theorem~\ref{thm:Damaschke}. 
\item[\family{7,22}{49}{Trivial}]
This class is included in paths with maximum independent set of size at most 2, that is paths of length at most 4.
\item[\family{3,16}{50}{2-Stars}]
This class is included into co-chordal $\cap$ forests. 
As a $P_5$ has a $C_5$ in its complement, we know that co-chordal graphs graphs have diameter at most 3. 
Thus the class is included into the 2-stars. 
Now, let $x$, $y$ be the two centers of a 2-star, and $x_1, ...x_p$ be the other neighbors of $x$, and $y_1, ..., y_q$ be the other neighbors of $y$.
The vertex ordering with first the isolated nodes, and then $x_1, ...x_p,  y_1, ..., y_q, x, y$ avoids the two patterns.
\item[\family{11,16}{51}{Trivial}] 
This class is included in forest $\cap$ co-forest  which is trivial as we already argued in the proof of Theorem~\ref{thm:complement}.
\item[\family{3,6,16}{52}{Stars}]
This class is included in the one of $[6,16]_{41}$, which is the stars. 
For the other inclusion, an ordering first the leaves, then the isolated nodes, and then the center avoids the three patterns.
\item[\family{3,5,16}{53}{Stars}]
Using Figure~\ref{fig:P4}, no ordering of a $P_4$ avoids the three patterns. We are left with trees of diameter at most 2, that is, stars, and the patterns allow only one non-trivial component. The ordering with first the isolated nodes, then the leaves then the center avoids the patterns.  
\item[\family{3,5,6,16}{54}{Stars (without isolated nodes)}]
The inclusion into stars follows from the cases $[3,6,16]_{52}$, and the isolated nodes are forbidden by Patterns~3, 5 and~6. The ordering with the leaves first avoid the patterns.
\item[\family{3,4,16}{55}{Trivial}]
As a subclass of case $[3,16]_{50}$, this class is included into 2-stars, and as a subclass of $[4,16]_{46}$ it is also in linear forests.
Then only a star with less than 2 leaves satisfies the condition. Therefore the class  is trivial.
\item[\family{4,11,16}{56}{Trivial}]
As a subclass of case $[11,16]_{51}$.
\item[\family{3,14,16}{57}{Trivial}]
A graph which is a forest and split graph is necessarily a 2-star. But as can be seen in  Figure~\ref{fig:P4} there is no ordering of a $P_4$ avoiding the three patterns. Therefore only stars are allowed. If we suppose that the graph contains a star with 3 leaves $a_1, a_2, a_3$ and a center $x$. Pattern~16 forbids $x$ to have 2 $a_i's$ in its right. Pattern~14 forbids $x$ to have 2 $a_i's$ in its left. Since in any ordering $x$ has at least 2 $a_i's$ in its left or its right, therefore there is no ordering avoiding the three patterns. So only a star with less than 2 leaves satisfies the condition. Therefore the class  is trivial.

\item[\family{11,14,16}{58}{Trivial}]
As a subclass of case $[11,16]_{51}$.
\item[\family{3,22}{59}{Trivial}]
This class is included in co-chordal paths. 
But a path with 5 vertices has a 4-cycle in its complement. Therefore only $P_i$ with $i\leq 4$  are allowed. 
\item[\family{3,6,22}{60}{Trivial}]
As a subclass of case $[3,22]_{59}$.
\item[\family{2,16}{61}{Stars}]
Checking with Figure~\ref{fig:P4}, the diameter has to be bounded by two, and trees of diameter 2 are stars. 
 The ordering with the leaves, then the centre and then the isolated nodes avoids both patterns.
\item[\family{2,7,16}{62}{Trivial}]
As subclass of case $[7,16]_{40}$.
\item[\family{10,16}{63}{Stars}]
As Pattern~2 extends Pattern~10, and as $[2,16]_{61}$ characterizes the stars, this classes is also included in stars. The same ordering as for $[2,16]_{61}$ works. 
\item[\family{7,10,16}{64}{Trivial}]
As subclass of case $[7,16]_{40}$.
\item[\family{2,5,16}{65}{Stars}]
The inclusion and ordering follow from $[2,16]_{61}$.
\item[\family{2,15,16}{66}{Trivial}]
The class is included in stars because of $[2,16]_{61}$. Suppose the star has three non-center nodes, then for any ordering, Pattern~15 appears. Thus there the star has degree at most two.
\item[\family{5,10,16}{67}{Stars}]
This class is included into stars using case $[10,16]_{63}$, and the same ordering as for $[2,16]_{61}$ works.
\item[\family{10,15,16}{68}{Trivial}]
As a subclass of case $[15,16]_{44}$.
\item[\family{2,4,16}{69}{Trivial}]
The class is included in stars because of $[2,16]_{61}$, but the patterns forbid a $P_3$, thus only one edge is possible.
\item[\family{2,4,7,16}{70}{Trivial}]
As subclass of case $[7,16]_{40}$.
\item[\family{2,22}{71}{Trivial}]
The class only contains paths because of Pattern~22, but no ordering of a $P_3$ can avoid both patterns, thus there is at most one edge.
\item[\family{2,7,22}{72}{Trivial}]
As subclass of case $[2,22]_{71}$.
\item[\family{24}{73}{Stars}] 
See Theorem~\ref{thm:Damaschke}.
\item[\family{7,24}{74}{Trivial}]
Because Pattern~16 extends Pattern~24, and $[7,16]_{40}$ characterizes a trivial class.
\item[\family{6,24}{75}{Stars}]
Included in stars, because of Pattern~24.
The ordering with the leaves, then the isolated nodes and then the center avoids the patterns.
\item[\family{19,24}{76}{Trivial}]
Pattern~16 extends Pattern~24, and $[16,19]_{42}$ defines a trivial class.
\item[\family{5,24}{77}{Stars}]
Included in stars, because of Pattern~24.
The ordering with the isolated nodes, the leaves and the center avoids both patterns.
\item[\family{15,24}{78}{Trivial}]
Pattern~7 extends Pattern~15, and $[7,24]_{74}$ defines a trivial class.
\item[\family{5,6,24}{79}{Stars (without isolated nodes)}]
The inclusion follows from cases $[6,24]_{75}$ and from the fact that the patterns forbid isolated nodes.
Having the center at the end is enough to avoid the patterns.

\item[\family{4,24}{80}{Trivial}]
Since Pattern~24 forces that only the last node can be the right-hand of an edge, and then Pattern~4 prevents this node to have more than one neighbor.
\item[\family{4,7,24}{81}{Trivial}]
As Pattern~16 extends Pattern~24, and the class of $[7,16]_{40}$ is trivial.
\item[\family{14,24}{82}{Trivial}]
Pattern~4 extends Pattern~14, and the class of $[4,24]_{80}$ is trivial. 
\item[\family{7,14,24}{83}{Trivial}]
As a subclass of case $[14,24]_{82}$.
\item[\family{18,24}{84}{Trivial}]
Pattern~4 extends Pattern~18, and the class of $[4,24]_{80}$ is trivial.
\item[\family{7,18,24}{85}{Trivial}]
As a subclass of case of $[18,24]_{84}$.
\item[\family{6,18,24}{86}{Trivial}]
As a subclass of case of $[18,24]_{84}$.
\item[\family{26}{87}{Trivial}]
See Theorem~\ref{thm:Damaschke}.
\end{description}

\subsection{Corollary of the proof}

From the proof we can extract the following results.

\begin{Coro}The following characterizations hold:
\begin{enumerate}
\item
Caterpillars can be defined via $\pforest$ \& co-$\pcomparability$. 
\item
Bipartite chain graphs via co-$\pchordal$ \& $\pbipartite$.
\item
Bipartite permutation graphs via $\ptriangle$ \& co-$\pcomparability$\\
or via $\pbipartite$ \& co-$\pcomparability$.
\item
Trivially perfect graphs via $\pchordal$ \& $\pcomparability$\\
or via $\pinterval$ \& $\pcomparability$.
\item
Threshold graphs via $\pchordal$ \& co-$\pinterval$\\
or via $\pcomparability$ \& $\psplit$
\end{enumerate}
\end{Coro} 

These results follow respectively from the cases
$[5,16]_{43}$, 
$[13,12]_{35}$, 
$[0,5]_{26}$,
$[5,12]_{33}$,
$[1,2]_{11}$, 
$[2,9]_{13}$,
$[1,10]_{12}$,
$[2,13]_{17}$
of the proof above.

\section{Algorithmic aspects}
\label{sec:algorithms}

We now consider the algorithmic implications of Theorem~\ref{thm:characterization}.
From a complexity point of view, we first note that for any pattern family (on any number of nodes), given the ordering it is possible to check in polynomial time if all the patterns are avoided. 
This means that the ordering is a polynomial certificate that a graph is in a class characterized by patterns, and thus the recognition of any such class is a problem in~NP.

For the case of patterns on three nodes, we first show that graph searches are especially well-suited for recognition of the classes, and then give a more exhaustive result about the complexity of the problem.
 
\subsection{Recognition via graph searches}

The ordering in which the nodes are visited during various graph searches has been studied in the literature, and is characterized using the so-called 4 points conditions on 3 vertices~\cite{CorneilK08}.
Such searches are essential tools for recognition of well-structured graph classes.
Therefore it is interesting to see how far we can go with graph searches to recognize the classes listed in Theorem~\ref{thm:characterization}. 
To our knowledge here is the state of the art in this direction.

\begin{Theo}\label{thm:graph-searches}
Generic search and BFS, DFS, LBFS, MNS and  Maximal Degree Search (MDS) can be used to obtain very simple algorithms for the recognition of the graph classes cited in Theorem~\ref{thm:characterization}.
\end{Theo}

\begin{proof}
\begin{description}
\item[Chordal.] If  $G$ is a chordal graph, any  LBFS provides a simplicial elimination ordering, i.e. an ordering avoiding the characteristic pattern  of chordal graphs \cite{RoseTL76}.
\item \textbf{Trivially perfect graphs.} As shown in \cite{Chu08} one LBFS is enough to recognize and certify trivially perfect graphs. Furthermore it produces  a characteristic vertex ordering.
 \item \textbf{Linear forest.} For each connected component just apply a BFS starting in a vertex $x$ and ending at $y$, then apply a second BFS starting at $y$.
If $G$ is a path, clearly the second BFS ordering avoids the forbidden pattern.
\item \textbf{Forest.} We have already observed in Theorem \ref{thm:Damaschke}, that if $G$ is a tree any generic search ordering provides an ordering avoiding the characteristic pattern.
\item \textbf{Star.} Perform a BFS starting at a vertex $x$ and ending at $y$.  Let $x_0$ the unique neighbour of $y$ (if $y$ has more neighbours then $G$ is not a star).
If $G$ is a star, any ordering of the vertices finishing with $x_0$ avoids the forbidden pattern.
\item
\textbf{Caterpillars.} Let us now consider a caterpillar $C$, using a result in \cite{CharbitHMN17} 2 consecutives BFS, the second one starting at the end vertex of the first one, provide on a caterpillar an ordering  avoiding  the forbidden patterns.  But they also provide a  diametral path, since it has been proved that for trees it yields a diametral path \cite{Handler73}. 
\item \textbf{2-star.} Same as for caterpillars with an extra test on the length of the diametral path that must be exactly 3.
\item \textbf{Bipartite.} Let us apply any Layered search on $G$, for example a BFS. Starting from $x_0$ the layers are $L_1, \dots L_k$.
Consider the ordering $\tau= x_0, L_2, \dots  L_{2p}, L_1, \dots L_{2p+1}$  when $k=2p+1$.
Clearly if $G$ is bipartite $\tau$ avoids the forbidden pattern.
\item \textbf{Proper interval.} If $G$ is a proper interval graph, a series of 3 consecutive LBFS$^+$  always produces an ordering avoiding the characteristic first 2 patterns of Theorem  \ref{thm:Damaschke} of proper interval graph \cite{Corneil04}.
\item \textbf{Interval.} If $G$ is an interval graph,  a series of 5 consecutive LBFS$^+$ followed by a special LBFS$^*$  always produces and ordering avoiding the characteristic pattern of interval graphs \cite{CorneilOS09}. It should also be noticed that a similar result can be obtained using another graph search, namely Maximal Neighbourhood Search, see \cite{LiW14}.
\item \textbf{Split.} Compute a maximal clique tree via a LBFS on $G$ (as explained in \cite{GalinierHP95}) and identify the clique of maximal size $C$. Do the same procedure on $\overline{G}$ (without building $\overline{G}$  just using $G$) and identify $I$ the clique of maximum size in $\overline{G}$. If $C$ and $I$ partition $V(G)$ then $G$ is a split graph. So with 2 LBFS we can construct an ordering that avoids the forbidden pattern.
\item \textbf{Threshold graphs.} A maximal degree search (MDS) which is a generic graph search that starts at a vertex of maximal degree and then break ties with the degrees, i.e., at each step the selected vertex is eligible and has a maximum degree in the remaining graph. 
If the graph $G$ is a threshold graph, this search generates an ordering $\tau$ of the vertices as described in Corollary \ref{coro:order-threshold}. Let $x_0$ the first vertex in $\tau$. To check if $G$ is a threshold graph it suffices to check if $N(x_0)$ is a clique and $\overline{N(x_0)}$ an independent set and that the neighbourhoods $\overline{N(x_0)}$ of are totally ordered with respect to their $\tau$ ordering. All of this can be done in linear time.
\item \textbf{Triangle-free.} Any ordering works, so we can use any graph search.
\item
\textbf{Triangle-free $\cap$ co-chordal.} A LBFS applied in the complement of the graph will provide an ordering of the vertices avoiding the 2 patterns.
\item
\textbf{Co-comparability.} If $G$ is a co-comparability  graph, a series of  $n$ consecutive LBFS$^+$ always produces an ordering avoiding the characteristic pattern  of  cocomparability  graphs \cite{DusartH17}.
\item \textbf{Permutation.}  In \cite{CorneilDHK16} a permutation recognition algorithm is presented which works as follows. First compute a cocomp ordering $\tau$ for $G$ and a cocomp ordering $\sigma$ for $\overline{G}$. Then transitively orient $G$ (resp. $\overline{G}$) using $\sigma $(resp. using $\tau$).
Then use a depth first search to compute the two orderings that represent the permutation graph, both of them avoids the patterns.
Using the above result for cocomparability graphs, thus in this case an ordering avoiding the patterns can be obtained via $2n+2$ consecutive graph searches.
\item
\textbf{Permutation Bipartite graphs.} Add to the recognition of permutation graphs a BFS to check if the graph is bipartite. 
\end{description}
\end{proof}

Therefore for each of the classes of Theorem~\ref{thm:characterization}, we can produce the characteristic ordering via a series of graph searches. Then using a brute force algorithm we can check if this ordering avoids the patterns in $O(n^3)$. Within this complexity we can also recognize complement classes.

\subsection{Complexity results}
As a consequence using Theorems~\ref{thm:characterization} and~\ref{thm:graph-searches}
we can recover the result of~\cite{HellMR14}.

\begin{Coro}(\cite{HellMR14})
All classes defined with sets of patterns on 3 vertices can be recognized in  $O(n^3)$.
\end{Coro}

\noindent Theorem~\ref{thm:characterization} actually allows us to get the following more fine-grain result.

\begin{Theo}\label{thm:linear}
All classes defined with sets of patterns on 3 vertices and their complements can be recognized in linear time, except triangle-free and comparability graphs.
\end{Theo}

\begin{proof}
Of course some of the classes of graphs in Theorem \ref{thm:characterization} such as stars, 2-stars, cliques, 1-splits, augmented cliques and complete bipartite graphs are quasi-trivial and therefore they are linear-time recognizable.

As detailed  in Theorem \ref{thm:graph-searches} the "classic" classes of graphs such as bipartite graphs, forests, linear forests, caterpillars, chordal graphs, interval graphs, proper interval graphs,  split graphs, permutation graphs can be recognized in linear time via a graph search (for example LexBFS). Therefore their complement classes can also be recognized in linear time, using the usual technique of partition refinement. 
Similar results also hold for trivially perfect graphs graphs \cite{Chu08} or threshold graphs and  bipartite chain graphs and their complement  \cite{HeggernesK07}.

The co-triangle-free $\cap$ chordal graph class can also be recognizable in linear time in the following way. First check if the graph is chordal and then compute its maximum independent set and check whether it has strictly more than 2 vertices. Both operations can be done in linear time see \cite{RoseTL76}. Similarly for their complement.

Their exists a linear time for the recognition of permutation graphs \cite{McConnellS99}. This algorithm produces two comparability orderings for the graph itself and its complement. It also gives a permutation representation of the graph which can be also tested in linear time. Then it suffices to check for the bipartitness, which can also be done in linear time. Similarly co-bipartitness can be checked in linear time and therefore complement can be also recognized  in linear time.
\end{proof}

Let us now discuss what is known about the complexity of the recognition of the two  remaining classes and their complement. For these classes the problem is not the computation of a good ordering but its certification. Up to our knowledge, no graph search can produce an ordering and its certification in linear time.

Let $\omega$ be the best exponent of the complexity of an algorithm for $n\times n$ boolean matrix multiplication. Using algorithm in \cite{Williams12} $\omega= 2,3727$.

First, triangle-free graphs (resp. co-triangle-free) can be recognized in $O(m^{1.41})$ (resp. $O(m^{1.186})$) for sparse graphs and in
$O(n^{2,3727})$ for dense ones. Indeed, the best known algorithm \cite{AlonYZ97} for recognition of triangle-free graphs is in $O(m^{\frac{2 \omega}{\omega+1}})= O(m^{1.41})$   if $\omega= 2,3727$. For dense graphs one may use boolean matrix multiplication
in $O(n^{2,3727})$. 

Recognition of triangle-free graphs is still an active area of research and nowadays some lower bounds under complexity hypothesis are discussed see \cite{WilliamsW18}.

For the recognition of  co-triangle-free graphs using a similar technique, it can be done in $O(m^{\frac{1}{2}\omega})=O(m^{1.186})$  \cite{DurandH11} for sparse graphs and using matrix multiplication for denses graphs.

Second for  comparability graphs, unfortunately the recognition algorithm LexBFS based presented in Theorem \ref{thm:graph-searches} has a worst-case complexity of $O(n \cdot m)$.Similalry for comparability graphs a vertex ordering that avoids its pattern \emph{if it is a comparability graph} can be computed in linear time \cite{McConnellS99}, but it is still not known if one can check if this ordering avoids the comparability pattern in linear time. 

But comparability graphs and their complement classes can be recognized in 
$O(n^{2,3727})$ using matrix multiplication.

The best known algorithm
is in $O(m^{\frac{2 \omega}{\omega+1}})= O(m^{1.41})$ \cite{HabibG18}. For dense graphs one can also use matrix multiplication. For dense cocomparability graphs up to our knowledge only the matrix multiplication has been proposed so far.

\begin{Coro}
All classes of graphs defined with sets of patterns on 3 vertices can be recognized in $O(n^{2,3727})$.
\end{Coro}


\section{Discussions and open problems}
\label{sec:discussions}

In this section, we discuss related topics, such as larger patterns, and review some open problems.

\subsection{Larger patterns}
In this paper we have focused on patterns on three nodes, and we have now good picture the the associated classes. 
An obvious next step is to look at larger patterns. 
We review a few topics in this direction.

\sparagraph*{Straight line patterns and colorings} \label{subsec:on-a-line}

To illustrate the expressivity of the pattern characterizations, we show here how to express colorability notions in terms of forbidden patterns. 

Let us define a notion of a colorability and a type of patterns.
First, a graph is \emph{$(a,b)$-colorable} \cite{DemangeEW05}, if one can partition the vertices into $a$ independent sets and $b$ cliques. 
In particular, classic $k$-colorability corresponds to $(k,0)$-colorability. 
Second, a \emph{straight line pattern} is a pattern where the decided edges are exactly the ones between consecutive vertices. 
For example, on three nodes, \pbipartite{} and \psplit{} are straight line patterns, because they are formed of an edge followed by respectively another edges and a non-edge, with the top edge being undecided. Also, on two nodes, the two patterns with an edge and a non-edge are straight line patterns. As noted in Subsection~\ref{subsec:def}, these define the independent sets and the cliques.

Note that on the four examples above there is a link between colorability and straight line patterns: the bipartite graphs, split graphs, independent sets and cliques are respectively the (2,0), (1,1), (1,0) and (0,1)-colorable graphs. This is actually a general phenomenon, as stated in the following theorem. 

\begin{Theo}\label{thm:coloring}
If $P$ is a straight line pattern with $a$ edges and $b$ non-edges, then the class~$\CC_{P}$ is the class of $(a,b)$-colorable graphs.
\end{Theo}

Before proving this theorem, let us highlight a remarkable property: the ordering of the edges and non-edges in the pattern does not matter. In other words, two straight line patterns with the same number of edges and non-edges define the exact same class.

\begin{proof}
We prove the result by induction on $a+b$.
For $a+b\leq 2$, we have already mentioned that the property is true (the case of $(0,2)$ follows from  $(2,0)$ by symmetry).
 
Now, let $P$ be an arbitrary straight line pattern with $a$ edges and $b$ non-edges, with $a+b>2$. 
It is sufficient to prove the result for the case when two last nodes of $P$ are linked by an edge. 
Indeed, if it is a non-edge, then we can consider the complement, and the result follows.
Let $P'$ be the same as $P$, but without the last vertex (and thus without the last edge). 
By induction, we assume that $P'$ satisfies the property.

We first show that the graphs of $\CC_P$ are $(a,b)$-colorable. 
Consider an ordered graph $G=(V,E)$ that avoids $P$. 
Let $X$ be the set of vertices of $G$ that have no neighbor to the right in the ordering. 
Note that $X$ is an independent set. 
We claim that the ordered subgraph $G'$ of $G$ induced by $V\setminus X$ avoids $P'$.
Indeed, if there is an occurrence of $P'$ in this subgraph then there is an occurrence of $P$ in $G$. This is because the last vertex of the occurrence has a neighbor on its right, and the pattern finishes by an edge.
Then as $G'$ avoids $P'$, by induction it is $(a-1,b)$-colorable. The fact that $X$ is an independent sets, leads to $G$ being $(a,b)$-colorable.

Conversely, let $G=(V,E)$ be $(a,b)$-colorable. 
Consider an independent set $X$ of the $(a,b)$-coloring.
Let $G'$ be the subgraph of $G$ induced by $V \setminus X$.
This subgraph is $(a-1,b)$-colorable. 
Thus by induction, it has an ordering that avoids $P'$.
Now, consider the ordering of $G$, made by concatenating this ordering of $G'$ and the nodes of $X$.
This ordering avoids~$P$. Indeed an occurrence of $P$ would imply an occurrence of $P'$ in the part of the ordering corresponding to $G'$.
\end{proof}

In the case $b=0$, Theorem~\ref{thm:coloring} corresponds to the classical Mirsky's Theorem, that states that the chromatic number of a graph $G$ is the minimum over all acyclic orientations of $G$ of the maximum length of a directed path \cite{Mirsky71}.

\begin{Coro}[Mirsky's Theorem]\label{cor:k-colorable}
A graph $G$ is $k$-colorable if and only if there exists an order on its vertices that avoids the straight line pattern with $k$ edges.
\end{Coro}

\sparagraph*{Complexity of recognition for larger patterns}

The previous paragraph has an important consequence: unlike patterns on three nodes, patterns on four nodes can define classes that are NP-hard to recognize. 
Indeed the straight line pattern with three edges defines the class of 3-colorable graphs, and 3-colorability is an NP-hard property~\cite{Karp72}.

On the other hand, it is not true that all patterns on four nodes define a class that is NP-hard to recognize. 
For example, outerplanar graphs can be recognized in linear time~\cite{Mitchell79}, and are defined by the pattern on four nodes of Figure~\ref{fig:outerplanar}.

\begin{figure}[!h]
\begin{center}
\begin{tikzpicture}   
[scale=1.0,auto=left,every node/.style={circle,draw,fill=black!5}]   
\node (n1) at (0,0) {1};   
\node (n2) at (1,0) {2};   
\node (n3) at (2,0)  {3};   
\node (n4) at (3,0)  {4};

\draw (n1) to[bend left=50] (n3); 
\draw (n2) to[bend left=50] (n4);
 
\end{tikzpicture}
\end{center}
\caption{\label{fig:outerplanar} The pattern that characterizes outerplanar graphs.}
\end{figure}
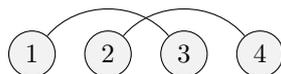

Finally, to illustrate the difficulty of having an intuition of the complexity of recognition, based only on the shape of the pattern, let us mention one more class.
Graphs of queue number 1~\cite{HeathR92}, are the ones that avoid the pattern of Figure~\ref{fig:queue-number-one}. 
These graphs are NP-hard to recognize~\cite{HeathLR92}, although the pattern is just a shuffling of the pattern of outerplanar graphs.

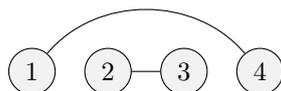
\begin{figure}[!h]
\begin{center}
\begin{tikzpicture}   
[scale=1.0,auto=left,every node/.style={circle,draw,fill=black!5}]   
\node (n1) at (0,0) {1};   
\node (n2) at (1,0) {2};   
\node (n3) at (2,0)  {3};   
\node (n4) at (3,0)  {4};

\draw (n1) to[bend left=50] (n4); 
\draw (n2) to[bend left=0] (n3);
 
\end{tikzpicture}
\end{center}
\caption{\label{fig:queue-number-one} The pattern that characterizes graphs of queue number 1.}
\end{figure}

Understanding the P/NP dichotomy for arbitrary patterns is clearly an essential problem here.
In this direction, Duffus et al \cite{DuffusGR95}  conjectured that most  classes characterized by  a 2-connected pattern are NP-complete to recognize.

\sparagraph*{Patterns on four nodes}

Before tackling arbitrary patterns, it seems reasonable to gather knowledge about patterns on four nodes. 
We already mentioned a few patterns on four nodes, for which we know the associated class.
Unfortunately, we do not know much more, although there are $3^6=729$ such patterns. 

\paragraph*{From patterns on three nodes}

Some patterns on four nodes are directly related to patterns on three nodes. 
For a pattern $Q$ on three nodes, one can consider the pattern $P$ on four nodes, made by adding a node on the right-end with only undecided edges to the nodes of~$Q$. Then the class $\CC_P$ is exactly the set of graphs that can be formed by taking a graph of $\CC_{Q}$ and (possibly) adding a vertex with an arbitrary adjacency. This follows from a proof similar to the one of Theorem~\ref{thm:coloring}.

\paragraph*{From forbidden subgraph characterizations}
Another way to reuse things we know is to exploit the remark at the end Subsection~\ref{subsec:basic}: a forbidden induced subgraph characterization can always be turned into a forbidden pattern characterization by taking all the orderings of this subgraph.
For example, the cographs as introduced in \cite{Seinsche74} are the $P_4$-free graphs, thus they are exactly the graphs that admit an ordering of the vertices avoiding the 12 patterns described in Figure~\ref{fig:P4}.
For cographs, there actually exists a more compact characterization using 2 patterns of size 3 and one of the $P_4$ orderings~\cite{Damaschke90}.
Other classes related to orderings of $P_4$ have been studied, like the perfectly orderable graphs~\cite{Chvatal84}, see~\cite{Brandstadt1999}.

\paragraph*{Intersection graphs}

There is one more specific class with a pattern on four nodes that we are aware of. This is the class of p-box graphs~\cite{SotoC15}, also known under other names in~\cite{CardinalFMTV18} and~\cite{CorreaFPS15}. 
Roughly, a graph is p-box if it is the intersection graph of rectangles having a corner on a line.
These are characterized by the pattern of Figure~\ref{fig:pbox}

\begin{figure}[h!]
\begin{center}
\begin{tikzpicture}   
[scale=1.0,auto=left,every node/.style={circle,draw,fill=black!5}]   
\node (n1) at (0,0) {1};   
\node (n2) at (1,0) {2};   
\node (n3) at (2,0)  {3};   
\node (n4) at (3,0)  {4};

\draw (n1) to[bend left=50] (n3); 
\draw (n2) to[bend left=50] (n4);
\draw[dashed] (n2) to (n3);
 
\end{tikzpicture}
\end{center}
\caption{\label{fig:pbox} The pattern that characterizes p-Box graphs.}
\end{figure}
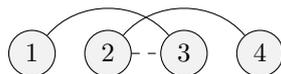

Note that this pattern is an extension of the pattern for outerplanar graphs, and therefore the p-box graphs contain the outerplanar graphs.
More generally a geometric approach is fruitful when looking at the extensions of the outerplanar pattern~\cite{FeuilloleyH19}.

\subsection{Graph parameters based on patterns}

For each pattern $P$, it is possible to define a graph parameter, that we call \emph{$P$-number}. 
One could also define the $\FF$-number  for a family $\FF$. Consider an ordered graph whose edges are colored.
We say that the coloring avoids the pattern, if for each color, the ordered graph induced by the edges of this color avoids the pattern. The $P$-number of a graph is the minimum number $k$ such that there exists an ordering of the graph and a coloring of its edges into $k$ colors, such that the coloring avoids the pattern~$P$.

The $P$-number has already been considered in the literature for some specific patterns $P$. 
For example, when $P$~is the pattern of outerplanar graphs (Figure~\ref{fig:outerplanar}), the $P$-number is known as the \emph{book thickness}~\cite{BernhartK79} or \emph{stack-number} \cite{DVW05}. 
Also for the pattern of Figure~\ref{fig:queue-number-one}, the parameter is the \emph{queue number}~\cite{HeathR92}.

Finally, if one fixes a number $k$ of colors, one can use the number of occurrences of the pattern as  another parameter. For the pattern of Figure~\ref{fig:outerplanar}, this is the \emph{$k$-page book crossing number}~\cite{ShahrokhiSSV96}.

One can view our main result as follows:

Recognition of  graph having $\FF$-number equal to 1, is polynomial (mostly linear) for every family of patterns on 3 vertices.
Therefore a natural question is: 
What is the complexity of recognizing $\FF$-number equal to 2, for the same families ?

\subsection{Ordering as a distributed certificate}

As noted in Section~\ref{sec:algorithms}, the recognition of the classes defined by patterns is always a problem in NP, as a correct ordering is a certificate that can be checked in polynomial time. 
This is notable, as it means that for all these classes there is a `standard' certificate. 

Certification mechanisms also appears in distributed computing. 
Namely in distributed decision~\cite{FeuilloleyF16, Feuilloley19}, one aims at deciding locally whether the network has some property. 
This often requires to help the nodes by giving them some piece of global information, that can be checked locally. This takes the form of a label for each node.
It has been noted that when the problem is to decide whether the network belongs to a class defined by a pattern (for example that the network is acyclic), providing each node with its rank in ordering is the optimal way to certify the property~\cite{Feuilloley18}.      

\subsection{Structural questions}

There are several interesting questions about the relations between classes defined by patterns.
For example, what can be said about the intersection of two classes defined by patterns?
The example of interval $\cap$ permutation shows that having two classes corresponding to patterns on three nodes does not imply that the intersection can be described by patterns on three nodes. 
On a related note, we emphasized in the text that in many cases the union-intersection property holds, but not always. 
It would be interesting to understand more precisely  in which case it does.
Another question is about pattern extension. When a patten is strictly included into another one, is it the case that the classes are also strictly included? It is the case when we restrict our attention to some specific pattern shapes~\cite{FeuilloleyH19}, but we have no general proof. Also note that when considering families, it is not true that adding a pattern reduces the family defined, as many cases Section~\ref{sec:main-proof} illustrate.

\vspace{1cm}
\vfill

\textbf{Acknowledgements:} We wish to thank the website \href{http://graphclasses.org/}{graphclasses.org}~\cite{graphclasses}, and its associated book~\cite{Brandstadt1999} that helped us to navigate through the jungle of graph classes. Furthermore  the authors wish to thank Yacine Boufkad, Pierre Charbit and Flavia Bonomo for fruitful discussions on this subject.
The first author is thankful to José Correa, Flavio Guinez and Mauricio Soto, for the very first discussions on this topic back in 2013.
Finally, we thanks the reviewers  for their thorough reading and their helpful comments.

\newpage

\DeclareUrlCommand{\Doi}{\urlstyle{same}}
\renewcommand{\doi}[1]{\href{https://doi.org/#1}{\footnotesize\sf doi:\Doi{#1}}}

\newcommand{\arxiv}[1]{{\footnotesize\sf arxiv: \href{http://arxiv.org/abs/#1}{#1}}}

\newcommand{\biburl}[1]{{\footnotesize\sf url: \href{#1}{#1}}}

\bibliographystyle{plainnat}
\bibliography{biblio-patterns}

\end{document}